\documentclass[11pt]{article}

\usepackage[T1]{fontenc}
\usepackage[utf8]{inputenc}
\usepackage{lmodern}
\usepackage{amsmath,amssymb,amsthm,mathtools}
\usepackage{booktabs,graphicx}
\usepackage[font=small,labelfont=bf,justification=raggedright,singlelinecheck=false]{caption}
\usepackage{flafter}
\usepackage{placeins}
\usepackage{enumitem}
\usepackage{etoolbox}
\usepackage{microtype}
\usepackage[a4paper,margin=27mm]{geometry}
\usepackage[hidelinks]{hyperref}
\usepackage{orcidlink}

\newcommand{\R}{\mathbb R}
\newcommand{\B}{\mathcal B}
\newcommand{\E}{\mathbb E}
\newcommand{\A}{\mathcal A}
\newcommand{\norm}[1]{\left\lVert#1\right\rVert}

\newtheorem{theorem}{Theorem}[section]
\newtheorem{proposition}[theorem]{Proposition}
\newtheorem{lemma}[theorem]{Lemma}
\newtheorem{corollary}[theorem]{Corollary}

\AfterEndEnvironment{proof}{\par\smallskip}

\title{Sequential Euclidean connections with exponential memory:\\
distributional performance and adversarial robustness}

\author{Pedro M. M. de Castro\,\orcidlink{0000-0002-9470-7458}\\
\small Centro de Inform\'atica, Universidade Federal de Pernambuco\\
\small Recife, Pernambuco, Brazil\\
\small \texttt{pmmc@cin.ufpe.br}}
\date{August 29, 2026}

\hypersetup{
  pdftitle={Sequential Euclidean connections with exponential memory: distributional performance and adversarial robustness},
  pdfauthor={Pedro M. M. de Castro}
}

\begin{document}

\maketitle

\begin{abstract}
Points in the unit ball of \(\R^d\) are processed sequentially.
Each new point \(p_i\) is connected to a state \(x_{i-1}\) that summarizes
earlier observations, after which
\(x_i=\gamma x_{i-1}+(1-\gamma)p_i\), with \(0\leq\gamma\leq1\).  The cost is
the sum of the \(\alpha\)-powers of the connection lengths.  This constant-gain
rule interpolates between the input-order path and the star centered at the
initial point.  For independent uniform points, we establish the stationary
insertion-length distribution and prove that it decreases in stochastic order
as \(\gamma\) increases.  If \(d+\alpha>2\), or if
\((d,\alpha)=(1,1)\), the optimal constant parameter satisfies
\(1-\gamma_N^*=\Theta(N^{-1/2})\), with an explicit asymptotic constant and
closed bounds.  For \(\alpha=1\), its leading expected tree length equals that
of the center star and is eventually smaller than the expected lengths of both
endpoint constructions.  For \(\alpha=2\), the optimizer is unique and
characterized exactly.  For the same \(N\), choosing \(1-\gamma_N\) as a fixed
positive multiple of \(N^{-1/2}\) gives a sharp two-term expansion of the
expected uniform-input cost and a maximal adversarial mean cost of
\(1+O(N^{-1/2})\).  For every fixed
\(0\leq\gamma<1\) and \(0<\alpha\leq3\), the exact asymptotic adversarial value
is \((2/(1+\gamma))^\alpha\).  When \(d\geq2\), exponential weighting is within
a factor smaller than \(1.161^\alpha\) of the best fixed nonnegative weighted
rule with the same average look-back, for \(0<\alpha\leq3\).  Comparison with
the running mean highlights its time-homogeneous update, stationary coefficient
profile, and fixed effective memory.
\end{abstract}

\medskip
\noindent\textbf{AI-use disclosure.}
OpenAI's ChatGPT and Codex were used as research tools to assist the author with
literature retrieval, mathematical exploration, symbolic and numerical checks,
and draft preparation and revision.  The author formulated the research
questions and mathematical framework, developed and refined the results and
arguments with this assistance, and critically reviewed the manuscript
throughout.  Responsibility for the mathematical statements and the final
manuscript rests entirely with the author.

\medskip
\noindent\textbf{Keywords.}
Geometric probability, stochastic recursion, exponentially weighted
averages, stochastic orders, online algorithms, worst-case analysis.

\smallskip
\noindent\textbf{Mathematics Subject Classification.}
60D05, 60E15, 60G10, 68W27.

\section{Introduction}

Consider a sequence of points in a Euclidean ball and a tree constructed in
the prescribed order of that sequence.  A central question is what information
about the preceding points should be retained when the next attachment point is
chosen.  The path in insertion order retains only the most recent point.  Under
random input, its expected cost is therefore governed by the distances between
successive points, and an adversary can force successive connections across a
diameter.  The star centered at the center \(O\) of the ball ignores the
observed sequence and keeps every edge within the radius, but its center is a
 prescribed anchor chosen independently of the input.  We ask whether a
 constant-gain state, generated only from the observations, can improve uniform-input
performance while reducing the worst-case scale of the path.  In this paper,
the distributional side means
expected performance for independent uniform points in the ball, while the
worst-case side allows an arbitrary prescribed input sequence.

Starting with \(x_0=p_0\), we attach \(p_i\) to \(x_{i-1}\) and then update
\[
 x_i=\gamma x_{i-1}+(1-\gamma)p_i,
\]
where \(0\leq\gamma\leq1\) is kept constant throughout the construction.  The
state \(x_i\) is the memory retained after processing \(p_i\): it summarizes the
entire observed sequence by giving progressively less weight to older points
and requires only one \(d\)-dimensional point of storage.  The parameter
\(\gamma\) controls how long the past influences subsequent attachments.  The
choice \(\gamma=0\) gives the path in insertion order, and \(\gamma=1\) gives
the star centered at \(p_0\).  For \(0<\gamma<1\), this is the geometric moving
average introduced by Roberts~\cite{Roberts1959}, now usually called an exponentially
weighted moving average.  The results below show that this constant-gain state
retains sequence-dependent attachment points, improves
the expected cost under uniform sampling, and reduces the exact worst-case
scale of the path.  The online rule stores one \(d\)-dimensional point of
 working state.  Retaining the updated points as labelled auxiliary vertices
 gives a caterpillar graph whose spine is \(x_0,x_1,\ldots,x_N\).  Our primary
 objective is the sum of the \(\alpha\)-powers of the
insertion lengths.
When \(\alpha=1\), this objective is also the Euclidean length of the subdivided
tree.  For other values of \(\alpha\), the edge-power cost of the labelled
 tree is related to the insertion cost by the explicit factor in
 Eq.\eqref{eq:tree-power-cost}.  Since this factor depends on \(\gamma\), a
 parameter minimizing the insertion cost need not minimize the edge-power cost.
 The optimization objective throughout is the insertion cost.

One concrete algorithmic motivation comes from sequential point location in a
fixed planar Delaunay triangulation.  Suppose that \(p_1,\ldots,p_N\) are
successive queries and that a triangle containing \(x_{i-1}\) is known.  A
straight walk from \(x_{i-1}\) to \(p_i\) visits the triangles intersected by
the segment \([x_{i-1},p_i]\).  The insertion length is therefore the length of
the geometric trajectory used by this walk.  Moreover,
\[
 x_i=x_{i-1}+(1-\gamma)(p_i-x_{i-1})
\]
lies on the same segment.  While locating \(p_i\), one can record a triangle
containing \(x_i\) and use it as the starting triangle for the next query.  The
same walk then performs two related tasks: it locates the current query and
provides the location state from which the following query can be processed.
In particular, the updated state requires no separate point-location walk.

This connection also explains the role of the memory parameter.  Starting each
walk at the preceding query makes the procedure immediately responsive to the
observed sequence, but alternating queries between distant regions produces
long displacements.  Starting every walk at a prescribed center gives a stable
reference point while discarding the spatial coherence of the queries.  The
state \(x_i\) lies between these two choices: it follows a persistent change in
the query locations and attenuates the effect of an isolated long jump.  Thus
the insertion objective measures the motion of a starting point that adapts to
the observed sequence with a tunable amount of memory.  Walking strategies for
point location are studied in~\cite{DevillersPionTeillaud2002}, and the use of
previous queries in distribution-sensitive point location is developed
in~\cite{CastroDevillers2013}.  Intersections with random geometric objects and
 navigation in random Delaunay triangulations are analyzed under specific
 probabilistic and boundary assumptions in
\cite{BoseDevroye1998,DevroyeMuckeZhu1998,BroutinDevillersHemsley2016}.
For a straight walk, the Euclidean insertion length is the length of its
trajectory.  The number of crossed triangles also depends on the geometry and
spatial distribution of the triangulation, so a traversal-complexity analysis
requires additional hypotheses.

The same recursion also has a restrained facility-location interpretation.  If
\(x_{i-1}\) is the position of a mobile facility and \(p_i\) is the next
request, the update moves the facility a fixed fraction \(1-\gamma\) toward
that request.  Both the travelled distance and the remaining service distance
are fixed multiples of the insertion length.  Related online models charge both
movement and service distances~\cite{FeldkordKnollmannMeyer2022}.  Smoothed
online optimization formalizes the sum of a movement cost and a per-round
hitting cost; Online Balanced Descent chooses each decision by balancing these
two terms~\cite{ChenGoelWierman2018,GoelLinSunWierman2019}.  Convex-function
chasing likewise asks an online decision maker to control function and movement
costs~\cite{ArgueGuptaGuruganesh2020}.  The rule studied here is a specialized
geometric policy: its gain is fixed in advance, and the new request determines
both the insertion vector and the next affine state.  This restricted form
permits exact distributional and adversarial calculations.

The first main result establishes the distributional side of the tradeoff.
For independent points uniformly distributed in the ball, let \(Z_\gamma\) be
the stationary insertion length.  We prove that
\[
 0\leq\gamma_1<\gamma_2<1
 \quad\Longrightarrow\quad
 \Pr(Z_{\gamma_2}>r)\leq\Pr(Z_{\gamma_1}>r)
 \quad(r\geq0).
\]
Thus every positive moment of the insertion length is nonincreasing as the rule gives
more weight to the past.  In the high-weight limit \((\gamma\to1^-)\), the attachment point
concentrates at \(O\), and the expected \(\alpha\)-power approaches
\(d/(d+\alpha)\), the corresponding value for the center star.  The proof uses
majorization and a theorem of Olkin and Tong~\cite{OlkinTong1988} on the peakedness of
weighted sums of independent random vectors with symmetric
log-concave densities.

The same phenomenon leads to a usable parameter choice for a fixed number of
insertions.  We optimize over rules whose value of \(\gamma\) remains constant
during the construction.  When \(d+\alpha>2\), as well as in the separate case
\((d,\alpha)=(1,1)\), every minimizing choice satisfies
\(1-\gamma_N^*=\Theta(N^{-1/2})\), and the expected cost is
\(Nd/(d+\alpha)\) plus an explicitly bounded term of order \(\sqrt N\).  The
constant is given by a one-dimensional integral and bounded above and below by
closed expressions in \(d\) and \(\alpha\).  For \(\alpha=2\), the complete
formula is explicit for every \(N\), and, for \(N\geq2\),
\[
 1-\gamma_N^*=N^{-1/2}-\tfrac12N^{-1}+O(N^{-3/2}).
\]
 The same scale also links the distributional and adversarial criteria for a
 fixed number \(N\) of insertions.  If \(\gamma_N=1-\lambda N^{-1/2}\), the uniform-input cost has
the expansion above with its \(\lambda\)-dependent constant, and the largest
mean cost over length-\(N\) input sequences is \(1+O(N^{-1/2})\).  The latter
bound includes the potential terms that remain when \(\gamma\) varies with
\(N\).
For ordinary Euclidean length, the comparison is especially direct.  The
input-order path and the star centered at \(p_0\) both have expected length
\(N\mu_d\), where \(\mu_d\) is the mean distance between two independent
uniform points.  The optimized exponential rule has leading expected cost
\(Nd/(d+1)\), which is the leading cost of the center star and is strictly
smaller than \(N\mu_d\).

The second main result supplies the worst-case side of the tradeoff.  For every
fixed \(0\leq\gamma<1\) and \(0<\alpha\leq3\), the largest asymptotic mean
insertion cost over arbitrary point sequences is
\[
 \left(\frac2{1+\gamma}\right)^\alpha.
\]
The lower bound is attained by alternating the endpoints of a diameter, and
the upper bound follows from explicit state potentials.  As \(\gamma\)
increases to one, this value approaches the unit worst-case scale of the center
star.  The rule therefore preserves an attachment point computed from the
input while its worst-case mean cost approaches that stable reference value.  For
 every fixed \(0<\gamma<1\) and integer \(m\geq2\), the periodic input that
 repeats one endpoint \(m\) times and then its antipode \(m\) times eventually
 has a strictly larger adversarial mean cost than antipodal alternation as
 \(\alpha\) tends to infinity.

 We also compare exponential weighting with all fixed nonnegative weighted
 averages for which the contributing points lie, on average, the same number of
 insertion steps before the most recent one.  In
dimension at least two and for \(0<\alpha\leq3\), its adversarial cost is less
than \(1.161^\alpha\) times the smallest cost in this class.  The ratio tends
 to one as this average grows.  This gives a direct justification
for exponential weights under that precise constraint.  The broader class
may have infinite support and need not admit a constant-state implementation.

Keeping \(\gamma\) constant is an essential restriction.  The arithmetic mean
of all observed points also stores one point, but its coefficients change at
 every insertion and its average number of insertion steps into the past grows
 with the sequence.  For uniform
 input, the difference between its expected cost and \(Nd/(d+\alpha)\) is of
 order \(\log N\).  Under the hypotheses of
 Theorem~\ref{thm:global-finite-optimizer}, this is smaller than the
 \(\sqrt N\) correction of the best constant \(\gamma\).  For
 \(0<\alpha\leq2\), its asymptotic adversarial mean cost is also one.  These
 facts delimit the role of exponential weighting precisely.  The fixed
 exponential rule is distinguished by its stationary behavior,
 time-homogeneous constant-gain update, fixed effective memory, and exact
 adversarial guarantees through \(\alpha=3\).

 The stationary monotonicity and second-moment formulas used in the central
 comparison are proved in the main text.  Further properties are organized by
 function in the appendices: a stochastic comparison for the running mean,
 exact transforms and fourth moments, high-memory distributional limits, and
 the periodic obstruction beyond \(\alpha=3\).  The stationary analysis belongs
 to the theory of iterated random functions~\cite{DiaconisFreedman1999}.

Sequential connection in Euclidean space provides the geometric background
\cite{Steele1989,CastroDevillers2011}.  Online Steiner tree problems study
related incremental connections in graphs and metrics
\cite{ImaseWaxman1991,GuptaKumar2014}, including recent variants with
 predictions~\cite{XuMoseley2022}.  Our earlier work
\cite{CastroDevillers2011} studies power-weighted Euclidean minimum
\(k\)-insertion trees under worst-case and uniform input, together with star
comparisons under uniform input; its case \(k=1\) is the input-order path used
here as an endpoint.  The present model replaces the terminal
attachment point by one affine summary of the observed sequence.  Exponentially
 weighted moving averages provide the recursion
\cite{Roberts1959}, while iterated random functions provide
the stationary framework~\cite{DiaconisFreedman1999}.  The results below join
these ingredients through uniform-input optimization and exact adversarial
analysis.  Section~\ref{sec:model} defines the rule
and its geometric realization.  Section~\ref{sec:uniform} treats uniform input,
the optimal constant parameter, and the time-varying comparison.
 Section~\ref{sec:adversarial} proves the worst-case results and compares fixed
 weighted averages.  The appendices collect the logically separate supporting
  results described above.

\paragraph{\textbf{Contributions.}}
The established ingredients used here include the exponential recursion,
stationary affine series, multivariate majorization, characteristic-function
products, and triangular-array central limit theory.  The contributions of the
paper are the radial finite-sample bounds for the insertion cost; the sharp
square-root scale and explicit constant for the optimal fixed gain; the exact
 quadratic optimizer for every \(N\); stochastic monotonicity of the stationary
 insertion length; a distributional and adversarial guarantee for the same \(N\); the
exact adversarial value through cubic costs and its periodic obstruction at
higher powers; and the approximation guarantee among fixed nonnegative
weighted averages.  The numerical section evaluates the explicit constants,
 \(N\)-dependent calibration, and finite-support weight comparison with a fully
reproducible implementation.

\section{The insertion strategy}
\label{sec:model}

Let
\[
 \B=\{p\in\R^d:\norm p\leq1\}
\]
be the unit ball, with center \(O\).  Throughout,
\(d\in\{1,2,\ldots\}\).  Let
\(p_0,p_1,\ldots,p_N\in\B\) be given in their insertion order.  The state is
required to be generated from the observed points, so the main model starts
with \(x_0=p_0\).  For a fixed value of \(\gamma\), define
\begin{equation}
 x_i=\gamma x_{i-1}+(1-\gamma)p_i,
 \qquad 0\leq\gamma\leq1.
 \label{eq:recursion}
\end{equation}
The insertion cost of the next \(N\) points is
\begin{equation}
 L_{\alpha,N}^{(\gamma)}
 =\sum_{i=1}^N\norm{p_i-x_{i-1}}^\alpha ,
 \qquad \alpha>0.
 \label{eq:insertion-cost}
\end{equation}
We call the construction defined by Eq.\eqref{eq:recursion} the
 \(\gamma\)-strategy.  Its state is a geometrically weighted average of the
 observed points and is stored as one point of \(\B\).  For fixed \(d\), this
 is constant-dimensional state in the real-RAM model.  The insertion index and
 the fixed parameter \(\gamma\) require only scalar storage.
 This statement concerns the state needed to compute the next attachment point.

\begin{proposition}
\label{prop:geometry}
For every input sequence in \(\B\), the recursion keeps \(x_i\) in \(\B\).
More precisely, \(x_i\) belongs to
\(\operatorname{conv}\{p_0,\ldots,p_i\}\).
With \(x_0=p_0\), the choice \(\gamma=1\) gives the star centered at
\(p_0\), and \(\gamma=0\) gives the path
\(p_0,p_1,\ldots,p_N\) defined by the insertion order.
\end{proposition}

\begin{proof}
Eq.\eqref{eq:recursion} is a convex combination of \(x_{i-1}\) and \(p_i\).
Induction gives the stated convex-hull property and hence \(x_i\in\B\).
For \(\gamma=1\), \(x_i=p_0\) at every step, so the segments are
\([p_0,p_i]\).  For \(\gamma=0\), \(x_i=p_i\), so the segments are
\([p_{i-1},p_i]\).
\end{proof}

The initialization rule excludes an exogenous anchor.  If the center \(O\) is
made available as an additional initial vertex, then \(x_0=O\) and
\(\gamma=1\) join \(p_1,\ldots,p_N\) to \(O\).  Connecting the point \(p_0\)
in that different model requires the additional segment \([O,p_0]\).  Under
uniform input, the known center is useful as a reference construction, while all
optimization results below retain the admissible initialization \(x_0=p_0\).

 For a precise geometric realization, retain a labelled auxiliary vertex \(x_i\)
 after the \(i\)-th update and subdivide \([x_{i-1},p_i]\) at \(x_i\).  It is a
 Steiner vertex of the geometric realization whenever it is distinct from the
 input points.  The abstract
graph has the edges
\[
 [x_{i-1},x_i]\quad\hbox{and}\quad[x_i,p_i],
 \qquad 1\leq i\leq N.
\]
Thus \(x_0,x_1,\ldots,x_N\) form a path, and each input point \(p_i\),
\(1\leq i\leq N\), is joined to \(x_i\).  This labelled abstract graph is the
tree used throughout the paper.  Geometric coincidence does not identify
distinct labels, and crossings or overlaps of straight segments do not create
additional vertices.  Edge lengths are counted with multiplicity.  A
zero-length edge may be contracted only by identifying its two adjacent
 endpoints.  Contracting the zero-length leaf edges when \(\gamma=0\) gives the
input-order path, while contracting the zero-length spine edges when
\(\gamma=1\) gives the star centered at \(p_0\).

Figure~\ref{fig:strategy-geometry} shows the two endpoint constructions and an
intermediate value of \(\gamma\) for the same ordered points.  The intermediate
 construction retains the spine \(x_0,x_1,\ldots,x_N\) of the caterpillar graph
 and joins each input point \(p_i\) to the auxiliary vertex \(x_i\) obtained
 after its insertion.

\begin{figure}[tbp]
\centering
\includegraphics[width=0.98\textwidth]{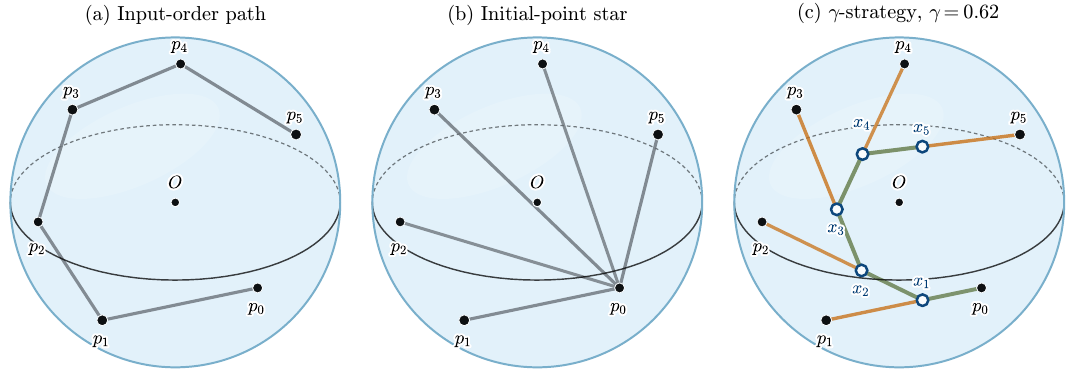}
\caption{Geometric constructions for the same ordered input.  Each panel uses
an orthographic projection of a three-dimensional scene.  The translucent
light-blue shell is the unit sphere, the thin black ellipse is a great circle,
and \(O\) is its center.  In panel (c),
black disks are the input points, open dark-blue circles are the auxiliary
vertices \(x_i\), green segments form the spine
\(x_0,x_1,\ldots,x_N\) of the caterpillar graph, and orange segments are the
leaf edges \([x_i,p_i]\).  Here \(x_0=p_0\).  The zero-length edges at
\(\gamma=0\) and \(\gamma=1\) are contracted in panels (a) and (b),
respectively.}
\label{fig:strategy-geometry}
\end{figure}

Since
\[
 \norm{x_i-x_{i-1}}=(1-\gamma)\norm{p_i-x_{i-1}},
 \qquad
 \norm{p_i-x_i}=\gamma\norm{p_i-x_{i-1}},
\]
the sum of the \(\alpha\)-powers of all edge lengths in the subdivided tree is
\begin{equation}
 \sum_{i=1}^N
 \left(\norm{x_i-x_{i-1}}^\alpha+\norm{p_i-x_i}^\alpha\right)
 =\bigl((1-\gamma)^\alpha+\gamma^\alpha\bigr)
  L_{\alpha,N}^{(\gamma)}.
 \label{eq:tree-power-cost}
\end{equation}
Thus the Euclidean tree length is \(L_{1,N}^{(\gamma)}\).  For
\(\alpha\neq1\), Eq.\eqref{eq:tree-power-cost} gives the edge-power cost of the
tree defined above for each fixed \(\gamma\).  This quantity
depends on the subdivision: dividing a segment of length \(r\) into \(m\)
equal parts changes its contribution from \(r^\alpha\) to
\(m^{1-\alpha}r^\alpha\).  The vertices \(x_i\) are specified by the
construction, so Eq.\eqref{eq:tree-power-cost} describes the tree obtained from
those specified subdivisions.  The optimization objective throughout is the
insertion cost \(L_{\alpha,N}^{(\gamma)}\).

All results are stated for the unit ball.  Scaling the ball by a factor \(R\)
multiplies every insertion cost of order \(\alpha\) by \(R^\alpha\).

The endpoint \(\gamma=1\) is singular for asymptotic statements.  For each
finite \(N\), Eq.\eqref{eq:recursion} and Eq.\eqref{eq:insertion-cost} remain
well defined and give the star centered at \(p_0\).  When \(0\leq\gamma<1\),
the influence of the initial state decays geometrically and the recursion has
a unique stationary law under independent uniform input.  At \(\gamma=1\),
the state stays at its initial value and no unique stationary law exists.  The
same distinction appears under adversarial input: if that model is extended
to \(\gamma=1\), then
\[
 \lim_{n\to\infty}\frac1n
 \sup_{p_1,\ldots,p_n\in\B}
 \sum_{i=1}^n\norm{p_i-x_0}^\alpha
 =(1+\norm{x_0})^\alpha.
\]
The values for fixed \(\gamma<1\) obtained in
Theorem~\ref{thm:through-three} tend to one as \(\gamma\) increases to one.
Thus taking the long-run limit before the endpoint limit gives a different
answer from fixing \(\gamma=1\) first.

\section{Uniform points in the ball}
\label{sec:uniform}

Throughout this section, \(p_0,p_1,\ldots\) are independent and uniformly
distributed in \(\B\), and \(0\leq\gamma<1\).

\subsection{Expected cost for a fixed number of insertions}
\label{sec:finite-sample}

Put
\[
 s_d=\E\norm{p_0}^2=\frac{d}{d+2},
 \qquad
 c_{d,\alpha}=\E\norm{p_0}^\alpha=\frac{d}{d+\alpha}.
\]
Fix a unit vector \(v\) and, for \(0\leq r\leq1\), define the radial mean
cost
\begin{equation}
 g_{d,\alpha}(r)
 =\frac1{\operatorname{vol}(\B)}
  \int_{\B}\norm{p-rv}^\alpha\,dp.
 \label{eq:radial-mean-cost}
\end{equation}
Rotational invariance makes this definition independent of the choice of
\(v\).  For \(x\in\R^d\) and \(t>0\), let
\[
 B(x,t)=\{y\in\R^d:\norm{y-x}\leq t\}.
\]
The next overlap comparison is a direct ball version of the translated-set
inequalities originating with Anderson~\cite{Anderson1955}.  The slice
proof also gives the strictness needed below.

\begin{lemma}
\label{lem:translated-ball-overlap}
For \(t>0\) and a unit vector \(v\), the function
\[
 r\longmapsto \operatorname{vol}\bigl(\B\cap B(rv,t)\bigr)
\]
is nonincreasing on \([0,\infty)\).  If
\(0\leq r_1<r_2\leq1\), then its value at \(r_2\) is strictly smaller than its
value at \(r_1\) for every \(t\) in a nonempty open interval containing one.
\end{lemma}

\begin{proof}
Slice both balls by lines parallel to \(v\).  On a line indexed by
\(z\in v^\perp\), the two slices, when nonempty, are intervals centered at
zero and \(r\), with half-lengths
\[
 \bigl(1-\norm z^2\bigr)^{1/2}
 \quad\hbox{and}\quad
 \bigl(t^2-\norm z^2\bigr)^{1/2}.
\]
The length of the intersection of two intervals with fixed half-lengths is
nonincreasing in the distance between their centers.  Integration over
\(v^\perp\) proves the first assertion.  For \(t=1\), the two central slices
have the same positive half-length and their overlap is strictly decreasing as
the separation increases from \(r_1\) to \(r_2\).  The same strict inequality
holds for all \(z\) in a neighborhood of zero and, by continuity, for all
\(t\) in a neighborhood of one.  Integration over this positive-measure set of
slices proves the second assertion.
\end{proof}

\begin{proposition}
\label{prop:conditional-kernel}
For a fixed \(x\in\B\) and an independent uniform point \(p\),
\begin{equation}
 \E\norm{p-x}^\alpha=g_{d,\alpha}(\norm x).
 \label{eq:conditional-moment}
\end{equation}
Moreover,
\begin{equation}
 g_{d,\alpha}(r)
 =c_{d,\alpha}+\frac\alpha2r^2+o(r^2)
 \quad\hbox{as }r\to0^+,
 \label{eq:radial-local-expansion}
\end{equation}
and \(g_{d,\alpha}(r)>c_{d,\alpha}\) for \(r>0\).
\end{proposition}

\begin{proof}
Rotational invariance shows that the conditional expectation depends on
\(x\) only through \(r=\norm x\), which proves
Eq.\eqref{eq:conditional-moment}.  Consider the same mean cost as a function of
the full state:
\[
 \frac1{\operatorname{vol}(\B)}
 \int_{\B}\norm{p-x}^\alpha\,dp.
\]
When \(d+\alpha>2\), the distributional second derivatives of
\(\norm{\cdot}^\alpha\) are locally integrable.  Fix \(\delta>0\).  If
\(p\in\B\) and \(x\) lies within distance \(\delta\) of \(\B\), then \(p-x\)
belongs to the ball of radius \(2+\delta\) centered at \(O\).  Each second
derivative is therefore integrable on the compact domain relevant to the
convolution.  Convolving it with the \(L^1\) indicator of \(\B\) gives a
continuous second derivative.  The resulting Hessian is continuous, and hence
bounded, on \(\B\).  Thus the mean-cost function is twice continuously
differentiable in a neighborhood of \(\B\).  Symmetry gives zero
first derivative at the origin.  The Hessian of the integrand at \(x=O\) is
\[
 \alpha\norm p^{\alpha-2}I
 +\alpha(\alpha-2)\norm p^{\alpha-4}pp^{\mathsf T}.
\]
Since the uniform distribution is isotropic and
\[
 \E\norm p^{\alpha-2}=\frac{d}{d+\alpha-2},
 \qquad
 \E\bigl(\norm p^{\alpha-4}pp^{\mathsf T}\bigr)
 =\frac1{d+\alpha-2}I,
\]
the expected Hessian is \(\alpha I\).  This proves
Eq.\eqref{eq:radial-local-expansion} whenever \(d+\alpha>2\).  In the remaining
one-dimensional cases, direct integration gives
\[
 g_{1,\alpha}(r)
 =\frac{(1+r)^{\alpha+1}+(1-r)^{\alpha+1}}
       {2(\alpha+1)},
\]
which has the same expansion.

For the strict inequality, the layer-cake representation gives
\[
 \int_{\B}\norm{p-rv}^\alpha\,dp
 =\int_0^\infty \alpha t^{\alpha-1}
   \left\{\operatorname{vol}(\B)
   -\operatorname{vol}\bigl(\B\cap B(rv,t)\bigr)\right\}\,dt.
\]
For \(0\leq r_1<r_2\leq1\),
Lemma~\ref{lem:translated-ball-overlap} makes the integrand at \(r_2\) at least
the integrand at \(r_1\) for every \(t\), with strict inequality on a nonempty
interval.  Integration in \(t\) therefore gives
\(g_{d,\alpha}(r_2)>g_{d,\alpha}(r_1)\).  Taking \(r_1=0\) proves the stated
inequality.
\end{proof}

For \(d\geq1\), put
\begin{equation}
 b_{d,\alpha}
 =\frac{d\,2^{d+\alpha-1}\Gamma(d/2)
          \Gamma((d+\alpha+1)/2)}
        {(d+\alpha)\sqrt\pi\,\Gamma(d+\alpha/2)}
  -\frac{d}{d+\alpha}.
 \label{eq:boundary-quadratic-constant}
\end{equation}

\begin{proposition}
\label{prop:radial-quotient}
Assume \(d\geq1\) and \(\alpha>0\).  Extend
\[
 q_{d,\alpha}(r)
 =\frac{g_{d,\alpha}(r)-c_{d,\alpha}}{r^2}
\]
to \(r=0\) by setting \(q_{d,\alpha}(0)=\alpha/2\).  If \(d\geq2\), this
function is strictly decreasing when \(0<\alpha<2\), is identically one when
\(\alpha=2\), and is strictly increasing when \(\alpha>2\).  If \(d=1\), it
is strictly increasing when \(0<\alpha<1\), is identically \(1/2\) when
\(\alpha=1\), is strictly decreasing when \(1<\alpha<2\), is identically one
when \(\alpha=2\), and is strictly increasing when \(\alpha>2\).  In every
case, \(q_{d,\alpha}(1)=b_{d,\alpha}\).
\end{proposition}

\begin{proof}
Take \(v\) to be the first coordinate vector.  First assume \(d\geq2\).  Let
\(\tau\) be the first coordinate of a uniform point on the unit sphere in
\(\R^d\), and let \(u\) be uniform on the unit sphere in \(\R^{d+2}\).  The
map \(p\mapsto\norm{p-rv}^{\alpha}\) belongs to \(W^{1,1}(\B)\): near its
only possible singularity the norm of its weak gradient is a constant times
\(\rho^{\alpha-1}\), whose radial integral is a constant times
\(\int_0^\varepsilon\rho^{d+\alpha-2}\,d\rho<\infty\).  The weak divergence
theorem is therefore applicable.  Integration by parts in the
first-coordinate density on the sphere gives, for \(0<r<1\),
\begin{align}
 g'_{d,\alpha}(r)
 &=-d\,\E\left[\tau(1+r^2-2r\tau)^{\alpha/2}\right] \notag\\
 &=\frac{d\alpha r}{d-1}\,
   \E\left[(1-\tau^2)(1+r^2-2r\tau)^{(\alpha-2)/2}\right] \notag\\
 &=\alpha r\,\E\norm{u-rv}^{\alpha-2}.
 \label{eq:radial-derivative-sphere}
\end{align}
Indeed, \(\E(1-\tau^2)=(d-1)/d\), and multiplication of the first-coordinate
density in dimension \(d\) by \(1-\tau^2\), followed by division by this
normalizing constant, gives the corresponding density in dimension \(d+2\).

If \(d=1\), direct integration gives
\[
 g'_{1,\alpha}(r)
 =\frac{(1+r)^\alpha-(1-r)^\alpha}{2}.
\]
The first coordinate of a uniform point \(u\) on the unit sphere in \(\R^3\)
is uniform on \([-1,1]\).  Direct integration in that coordinate therefore
gives
\[
 \alpha r\,\E\norm{u-rv}^{\alpha-2}
 =\frac{(1+r)^\alpha-(1-r)^\alpha}{2}.
\]
Thus Eq.\eqref{eq:radial-derivative-sphere} holds for every \(d\geq1\).

Let
\[
 m(r)=\E\norm{u-rv}^{\alpha-2}.
\]
For \(x\) in the open unit ball of \(\R^{d+2}\), the distance from \(x\) to
the unit sphere is positive.  The kernel and all its derivatives are then
uniformly bounded in a neighborhood of \(x\), which justifies differentiation
under the expectation and gives
\[
 \Delta_x\E\norm{u-x}^{\alpha-2}
 =(\alpha-2)(d+\alpha-2)
   \E\norm{u-x}^{\alpha-4}.
\]
This potential is radial.  Hence
\[
 \bigl(r^{d+1}m'(r)\bigr)'
 =r^{d+1}\Delta_x\E\norm{u-x}^{\alpha-2}\big|_{x=rv}.
\]
Since \(m'(0)=0\), the sign of \(m'(r)\) for \(r>0\) is the sign of
\((\alpha-2)(d+\alpha-2)\).  This gives the stated alternatives for both
\(d\geq2\) and \(d=1\).  Integrating
Eq.\eqref{eq:radial-derivative-sphere} and substituting \(s=rt\) yield
\begin{equation}
 q_{d,\alpha}(r)=\alpha\int_0^1 tm(rt)\,dt.
 \label{eq:radial-quotient-sphere}
\end{equation}
Thus \(q_{d,\alpha}\) has the asserted monotonicity.  Near \(u=v\) on the
sphere in \(\R^{d+2}\), the surface element and the kernel in \(m(1)\) give a
radial integrand of order \(\rho^{d+\alpha-2}\).  Since
\(d+\alpha-2>-1\), dominated convergence extends the conclusion to \(r=1\).
To evaluate the endpoint, use polar coordinates
\(p-v=\rho\omega\).  The condition \(p\in\B\) becomes
\(\omega_1\leq0\) and \(0\leq\rho\leq-2\omega_1\), so
\begin{align*}
 g_{d,\alpha}(1)
 &=\frac{1}{\operatorname{vol}(\B)}
   \int_{\omega_1\leq0}\int_0^{-2\omega_1}
        \rho^{d+\alpha-1}\,d\rho\,d\omega\\
 &=\frac{d\,2^{d+\alpha-1}\Gamma(d/2)
          \Gamma((d+\alpha+1)/2)}
        {(d+\alpha)\sqrt\pi\,\Gamma(d+\alpha/2)}.
\end{align*}
For \(d=1\), the same calculation uses the counting measure on the
zero-dimensional sphere.  Subtracting \(c_{d,\alpha}\) proves the last
assertion in every dimension.
\end{proof}

It follows that the best constants in a quadratic envelope for the conditional
cost are
\begin{equation}
\begin{aligned}
 k^-_{d,\alpha}
 &=\min_{0\leq r\leq1}q_{d,\alpha}(r)
   =\min\left\{\frac\alpha2,b_{d,\alpha}\right\},\\
 k^+_{d,\alpha}
 &=\max_{0\leq r\leq1}q_{d,\alpha}(r)
   =\max\left\{\frac\alpha2,b_{d,\alpha}\right\}.
\end{aligned}
 \label{eq:best-envelope-constants}
\end{equation}
In particular, \(0<k^-_{d,\alpha}\leq k^+_{d,\alpha}<\infty\), and
\begin{equation}
 c_{d,\alpha}+k^-_{d,\alpha}\norm x^2
 \leq \E\norm{p-x}^\alpha
 \leq c_{d,\alpha}+k^+_{d,\alpha}\norm x^2.
 \label{eq:best-quadratic-envelope}
\end{equation}
Neither constant in Eq.\eqref{eq:best-quadratic-envelope} can be improved
uniformly over \(x\in\B\).

\begin{lemma}
\label{lem:finite-state-second}
For \(t\geq0\), with \(x_0=p_0\),
\begin{equation}
 \E\norm{x_t}^2
 =s_d\left[
 \frac{1-\gamma}{1+\gamma}
 +\frac{2\gamma}{1+\gamma}\gamma^{2t}
 \right].
 \label{eq:finite-state-second}
\end{equation}
\end{lemma}

\begin{proof}
Iteration of the recursion gives
\[
 x_t=\gamma^tp_0+(1-\gamma)
      \sum_{j=1}^t\gamma^{t-j}p_j.
\]
The summands are independent and centered.  Hence
\[
 \E\norm{x_t}^2
 =s_d\left[\gamma^{2t}
 +(1-\gamma)^2\sum_{j=0}^{t-1}\gamma^{2j}\right],
\]
which reduces to Eq.\eqref{eq:finite-state-second}.
\end{proof}

\begin{theorem}
\label{thm:finite-envelope}
For \(\alpha>0\) and \(0\leq\gamma\leq1\),
\begin{equation}
 Nc_{d,\alpha}+k^-_{d,\alpha}
   \sum_{t=0}^{N-1}\E\norm{x_t}^2
 \leq \E L_{\alpha,N}^{(\gamma)}
 \leq Nc_{d,\alpha}+k^+_{d,\alpha}
   \sum_{t=0}^{N-1}\E\norm{x_t}^2.
 \label{eq:finite-envelope}
\end{equation}
The accumulated state second moment is
\begin{equation}
 \sum_{t=0}^{N-1}\E\norm{x_t}^2
 =s_d\left[
 N\frac{1-\gamma}{1+\gamma}
 +\frac{2\gamma}{1+\gamma}
  \frac{1-\gamma^{2N}}{1-\gamma^2}
 \right],
 \label{eq:finite-state-sum}
\end{equation}
Eq.\eqref{eq:finite-state-sum} extends continuously to \(\gamma=1\).
For \(\widehat\gamma_N=1-N^{-1/2}\), define the correction factor
\begin{equation}
 C_N=\frac{1}{2-N^{-1/2}}
 +\frac{2(1-N^{-1/2})
       [1-(1-N^{-1/2})^{2N}]}
       {(2-N^{-1/2})^2}.
 \label{eq:finite-CN}
\end{equation}
Then \(C_N\to1\), and the bounds for \(\widehat\gamma_N\) become
\begin{equation}
 Nc_{d,\alpha}+k^-_{d,\alpha}s_dC_N\sqrt N
 \leq \E L_{\alpha,N}^{(\widehat\gamma_N)}
 \leq Nc_{d,\alpha}+k^+_{d,\alpha}s_dC_N\sqrt N.
 \label{eq:finite-square-root-envelope}
\end{equation}
\end{theorem}

\begin{proof}
At insertion \(t+1\), the point \(p_{t+1}\) is independent of \(x_t\).
Apply Eq.\eqref{eq:best-quadratic-envelope}, take expectations, and sum
Eq.\eqref{eq:finite-state-second} for \(t=0,\ldots,N-1\).  Evaluating the two
geometric sums gives Eq.\eqref{eq:finite-state-sum} and proves
Eq.\eqref{eq:finite-envelope}.  Substituting
\(1-\widehat\gamma_N=N^{-1/2}\) into Eq.\eqref{eq:finite-state-sum} gives
\[
 \sum_{t=0}^{N-1}\E\norm{x_t}^2=s_dC_N\sqrt N.
\]
Eq.\eqref{eq:finite-CN} also gives \(C_N\to1\).
\end{proof}

The constants in Theorem~\ref{thm:finite-envelope} are obtained from the
one-variable radial mean cost and are optimal for a pointwise quadratic
envelope.  The exact second-order constant also has a geometric
one-dimensional representation.  Define
\begin{equation}
 H_{d,\alpha}
 =2\int_0^1\frac{1-r^d}{r}
   \left(g_{d,\alpha}(r)-c_{d,\alpha}\right)\,dr.
 \label{eq:H-constant}
\end{equation}
The local expansion Eq.\eqref{eq:radial-local-expansion} proves integrability at
the origin, and Proposition~\ref{prop:conditional-kernel} proves
\(H_{d,\alpha}>0\).  Eq.\eqref{eq:H-constant} involves only the
original uniform-ball mean cost.

\begin{proposition}
\label{prop:closed-square-root-bounds}
Assume \(d\geq1\) and \(\alpha>0\).  Then
\begin{equation}
 s_d\min\left\{\frac\alpha2,b_{d,\alpha}\right\}
 \leq H_{d,\alpha}\leq
 s_d\max\left\{\frac\alpha2,b_{d,\alpha}\right\}.
 \label{eq:closed-H-bounds}
\end{equation}
If \(d+\alpha>2\), the following weaker bounds involve only powers and
rational functions:
\begin{equation}
\begin{split}
 s_d\min\left\{\alpha2^{\alpha-3},
       \frac{\alpha d}{2(d+\alpha-2)}\right\}
 &\leq H_{d,\alpha}\\
 &\leq s_d\max\left\{\alpha2^{\alpha-3},
       \frac{\alpha d}{2(d+\alpha-2)}\right\},
\end{split}
\label{eq:power-H-bounds}
\end{equation}
For \((d,\alpha)=(1,1)\), one has \(H_{1,1}=1/6\).
\end{proposition}

\begin{proof}
Since \(g_{d,\alpha}(r)-c_{d,\alpha}=r^2q_{d,\alpha}(r)\),
\[
 H_{d,\alpha}
 =2\int_0^1r(1-r^d)q_{d,\alpha}(r)\,dr,
 \qquad
 2\int_0^1r(1-r^d)\,dr=s_d.
\]
Proposition~\ref{prop:radial-quotient} shows that
\(q_{d,\alpha}(r)\) lies between its endpoint values
\(q_{d,\alpha}(0)=\alpha/2\) and
\(q_{d,\alpha}(1)=b_{d,\alpha}\) throughout \([0,1]\).  The integral defining
\(H_{d,\alpha}\) is a weighted average of these values with total weight
\(s_d\), which proves
Eq.\eqref{eq:closed-H-bounds}.

For \((d,\alpha)=(1,1)\), Proposition~\ref{prop:radial-quotient} gives
\(q_{1,1}(r)=1/2\), and the preceding integral gives \(H_{1,1}=1/6\).
The remaining assertion concerns the case \(d+\alpha>2\).  For its proof, let
\[
 \Delta(t)=\Pr(\norm{p-v}>t)-\Pr(\norm p>t),\qquad 0\leq t\leq2.
\]
Lemma~\ref{lem:translated-ball-overlap} gives \(\Delta(t)\geq0\).  The
layer-cake formula and the identity \(b_{d,2}=1\) give
\[
 b_{d,\alpha}=\int_0^2\alpha t^{\alpha-1}\Delta(t)\,dt,
 \qquad
 1=\int_0^2 2t\Delta(t)\,dt.
\]
Comparing the ratio \((\alpha/2)t^{\alpha-2}\) with its value at \(t=2\)
shows that \(b_{d,\alpha}\geq\alpha2^{\alpha-3}\) for \(0<\alpha\leq2\)
and \(b_{d,\alpha}\leq\alpha2^{\alpha-3}\) for \(\alpha\geq2\).  Also,
\[
 \frac{\alpha d}{2(d+\alpha-2)}\geq\frac\alpha2
 \quad\hbox{for }\alpha\leq2,
 \qquad
 \frac{\alpha d}{2(d+\alpha-2)}\leq\frac\alpha2
 \quad\hbox{for }\alpha\geq2.
\]
Combining these comparisons with Eq.\eqref{eq:closed-H-bounds} proves
Eq.\eqref{eq:power-H-bounds}.
\end{proof}

\begin{lemma}
\label{lem:innovation-estimates}
Assume \(d+\alpha>2\), or \((d,\alpha)=(1,1)\), and let
\(\Lambda\) be a compact subset of \((0,\infty)\).  For
\(\lambda\in\Lambda\), put
\[
 \delta_N=\lambda N^{-1/2},\qquad
 \gamma=1-\delta_N,\qquad
 \eta_{t,N}=\delta_N\sum_{j=1}^t\gamma^{t-j}p_j .
\]
Uniformly for \(\lambda\in\Lambda\) and \(0\leq t<N\),
\begin{align}
 \E\norm{\eta_{t,N}}^2
 &=s_d\frac{\delta_N}{1+\gamma}(1-\gamma^{2t}),
 \label{eq:innovation-second-moment}\\
 \E\norm{\eta_{t,N}}^4&\leq C\delta_N^2.
 \label{eq:innovation-fourth-moment}
\end{align}
Let \(f(x)=g_{d,\alpha}(\norm x)\).  There is a constant \(C\), independent of
\(t,N\), and \(\lambda\in\Lambda\), such that
\begin{equation}
 \left|\E\{f(\gamma^tp_0+\eta_{t,N})-f(\gamma^tp_0)\}\right|
 \leq C\E\norm{\eta_{t,N}}^2.
 \label{eq:innovation-global-taylor}
\end{equation}
For every \(\varepsilon>0\), there are \(M>0\) and \(N_0\) such that, for
\(N\geq N_0\) and \(M\sqrt N\leq t<N\),
\begin{align}
 &\left|\E\{f(\gamma^tp_0+\eta_{t,N})-f(\gamma^tp_0)\}
 -\frac{\alpha}{2}\E\norm{\eta_{t,N}}^2\right|\notag\\
 &\hspace{35mm}\leq
 \frac{\varepsilon}{2}\E\norm{\eta_{t,N}}^2+C\delta_N^2 .
 \label{eq:innovation-local-taylor}
\end{align}
\end{lemma}

\begin{proof}
The summands defining \(\eta_{t,N}\) are independent and centered.  Their
covariance matrices are \(I/(d+2)\), and summing the squared coefficients gives
Eq.\eqref{eq:innovation-second-moment}.  For a partial sum
\(\sum_jb_jp_j\), expansion of the fourth power and independence give an upper
bound proportional to \((\sum_jb_j^2)^2\).  Here
\[
 \sum_jb_j^2
 =\delta_N^2\sum_{j=0}^{t-1}\gamma^{2j}
 \leq\frac{\delta_N}{1+\gamma}.
\]
This proves Eq.\eqref{eq:innovation-fourth-moment}, with a constant uniform on
\(\Lambda\).

The proof of Proposition~\ref{prop:conditional-kernel} shows that \(f\) is twice
continuously differentiable on \(\B\), has bounded Hessian there, and satisfies
\(D^2f(O)=\alpha I\).  In the case \((d,\alpha)=(1,1)\), these properties follow
 from \(f(x)=(1+x^2)/2\).  Condition on \(p_0\).  The vector
 \(\eta_{t,N}\) is independent of \(p_0\) and has conditional mean zero, so the
 linear term in Taylor's formula has conditional expectation zero.  The bounded
 Hessian and the integral remainder then prove
 Eq.\eqref{eq:innovation-global-taylor}.

 Choose \(r_0>0\) so that
\[
 \norm{D^2f(z)-\alpha I}\leq\varepsilon
 \quad\hbox{when }\norm z\leq r_0,
\]
and then choose \(M\) so that
\(\exp(-M\min\Lambda)\leq r_0/2\).  If \(t\geq M\sqrt N\), then
 \(\norm{\gamma^tp_0}\leq r_0/2\).  Put \(a_t=\gamma^tp_0\) and
 \(G=\{\norm{\eta_{t,N}}\leq r_0/2\}\).  Taylor's formula gives
 \[
  f(a_t+\eta_{t,N})-f(a_t)
  =\langle\nabla f(a_t),\eta_{t,N}\rangle
   +\int_0^1(1-s)\eta_{t,N}^{\mathsf T}
      D^2f(a_t+s\eta_{t,N})\eta_{t,N}\,ds.
 \]
 Conditioning again on \(p_0\), independence and centering cancel the linear
 term.  Subtracting
 \(\frac{\alpha}{2}\E\norm{\eta_{t,N}}^2\) leaves the expectation of the
 same integral with \(D^2f(a_t+s\eta_{t,N})-\alpha I\) in place of the
 Hessian.  On \(G\), the whole segment from \(a_t\) to
 \(a_t+\eta_{t,N}\) lies in the ball of radius \(r_0\), and hence the absolute
 value of this integral is at most
 \(\frac{\varepsilon}{2}\norm{\eta_{t,N}}^2\).  On \(G^c\), boundedness of
 the Hessian gives a constant multiple of \(\norm{\eta_{t,N}}^2\), while
\[
 \E\left[\norm{\eta_{t,N}}^2
 \mathbf1_{\{\norm{\eta_{t,N}}>r_0/2\}}\right]
 \leq\frac4{r_0^2}\E\norm{\eta_{t,N}}^4
 \leq C\delta_N^2.
\]
 Thus the contribution from \(G^c\) is at most \(C\delta_N^2\).  Combining
 the estimates on \(G\) and \(G^c\) proves
 Eq.\eqref{eq:innovation-local-taylor}.
\end{proof}

\begin{theorem}
\label{thm:square-root-constant}
Assume \(d+\alpha>2\), or \((d,\alpha)=(1,1)\).  Fix \(\lambda>0\) and set
\(\gamma_N=1-\lambda N^{-1/2}\) for all sufficiently large \(N\).  Then
\begin{equation}
 \lim_{N\to\infty}
 \frac{\E L_{\alpha,N}^{(\gamma_N)}-Nc_{d,\alpha}}{\sqrt N}
 =K_{d,\alpha}(\lambda)
 =\frac{\alpha s_d}{4}\lambda+\frac{H_{d,\alpha}}{2\lambda}.
 \label{eq:tight-square-root-constant}
\end{equation}
The convergence in Eq.\eqref{eq:tight-square-root-constant} is locally uniform
for \(\lambda\in(0,\infty)\).
The unique minimizing value and the corresponding constant are
\begin{equation}
 \lambda^*_{d,\alpha}
 =\left(\frac{2H_{d,\alpha}}{\alpha s_d}\right)^{1/2},
 \qquad
 K_{d,\alpha}(\lambda^*_{d,\alpha})
 =\left(\frac{\alpha s_dH_{d,\alpha}}{2}\right)^{1/2}.
 \label{eq:optimal-square-root-constant}
\end{equation}
For \(\alpha=2\), \(H_{d,2}=s_d\), so
\(\lambda^*_{d,2}=1\) and \(K_{d,2}(\lambda^*_{d,2})=s_d\).
Under the hypotheses of the theorem, the closed bounds are
\begin{equation}
 \min\left\{1,\sqrt{\frac{2b_{d,\alpha}}\alpha}\right\}
 \leq \lambda^*_{d,\alpha}\leq
 \max\left\{1,\sqrt{\frac{2b_{d,\alpha}}\alpha}\right\},
 \label{eq:closed-lambda-bounds}
\end{equation}
If \(d+\alpha>2\), the bounds are
\begin{equation}
 \min\left\{2^{(\alpha-2)/2},
       \sqrt{\frac d{d+\alpha-2}}\right\}
 \leq \lambda^*_{d,\alpha}\leq
 \max\left\{2^{(\alpha-2)/2},
       \sqrt{\frac d{d+\alpha-2}}\right\}.
 \label{eq:power-lambda-bounds}
\end{equation}
For \((d,\alpha)=(1,1)\), one has
\(\lambda^*_{1,1}=1\) and \(K_{1,1}(\lambda^*_{1,1})=1/6\).
Multiplication of the endpoints in Eq.\eqref{eq:closed-lambda-bounds}, and in
Eq.\eqref{eq:power-lambda-bounds} when applicable, by \(\alpha s_d/2\) gives the
corresponding bounds for \(K_{d,\alpha}(\lambda^*_{d,\alpha})\).
\end{theorem}

\begin{proof}
Let \(\delta_N=\lambda N^{-1/2}\), \(\gamma=1-\delta_N\), and let
\[
 x_t=\gamma^tp_0+\eta_{t,N},
 \qquad
 \eta_{t,N}=\delta_N\sum_{j=1}^t\gamma^{t-j}p_j.
\]
 Let \(f(x)=g_{d,\alpha}(\norm x)\) be the conditional expected cost at state
 \(x\).  Lemma~\ref{lem:innovation-estimates} supplies the moment estimates and
 the Taylor remainder needed below.  Its constants are uniform when
 \(\lambda\) ranges over a compact subset of \((0,\infty)\).

We next prove
\begin{equation}
 \lim_{N\to\infty}\frac1{\sqrt N}
 \sum_{t=0}^{N-1}
 \left\{\E f(x_t)-\E f(\gamma^tp_0)\right\}
 =\frac{\alpha s_d}{4}\lambda.
 \label{eq:innovation-contribution}
\end{equation}
 For fixed \(M\), the contribution of \(t<M\sqrt N\), divided by
 \(\sqrt N\), tends to zero by
 Eq.\eqref{eq:innovation-global-taylor} and
 Eq.\eqref{eq:innovation-second-moment}.  For the remaining indices,
 Eq.\eqref{eq:innovation-local-taylor} shows that the accumulated difference
 from \((\alpha/2)\sum_t\E\norm{\eta_{t,N}}^2\), divided by \(\sqrt N\), is
 \(O(\varepsilon)+o(1)\).  Finally, exact summation of
 Eq.\eqref{eq:innovation-second-moment} gives
\[
 \frac1{\sqrt N}\sum_{t=0}^{N-1}\E\norm{\eta_{t,N}}^2
 =\frac{s_d\delta_N}{(1+\gamma)\sqrt N}
 \left(N-\frac{1-\gamma^{2N}}{1-\gamma^2}\right)
 \longrightarrow\frac{s_d\lambda}{2}.
\]
Letting \(\varepsilon\to0^+\) proves
Eq.\eqref{eq:innovation-contribution}.

For completeness, the quadratic envelope gives
\[
 0\leq\E\left[g_{d,\alpha}(r\norm{p_0})-c_{d,\alpha}\right]
 \leq k^+_{d,\alpha}s_dr^2,
 \qquad 0\leq r\leq1.
\]
Moreover, \(-\sqrt N\log(1-\lambda N^{-1/2})\to\lambda\) uniformly when
\(\lambda\) ranges
over a compact subset of \((0,\infty)\), and
\[
 \gamma^t
 =\exp\left[\sqrt N\log(1-\lambda N^{-1/2})\frac{t}{\sqrt N}\right].
\]
On a bounded interval in
\(u=t/\sqrt N\), ordinary Riemann-sum convergence is therefore uniform in
such \(\lambda\).  On the remaining tail, the preceding quadratic bound supplies a
uniformly summable geometric majorant.  Consequently, the contribution
retained from the initial point satisfies
\begin{align*}
 &\lim_{N\to\infty}\frac1{\sqrt N}
 \sum_{t=0}^{N-1}
 \left\{\E f(\gamma^tp_0)-c_{d,\alpha}\right\}\\
 &\quad=\int_0^\infty
 \left\{\E g_{d,\alpha}(e^{-\lambda u}\norm{p_0})
                  -c_{d,\alpha}\right\}\,du
 =\frac{H_{d,\alpha}}{2\lambda}.
\end{align*}
For the last equality, use that \(\norm{p_0}\) has density
\(d r^{d-1}\) on \([0,1]\).  Tonelli's theorem
\cite[Sec.~18]{Billingsley1995} applies because the integrand
is nonnegative.  Substitution of \(s=e^{-\lambda u}r\), followed by interchange of
the two radial integrals, gives
\begin{align*}
 \int_0^\infty
   \left\{\E g_{d,\alpha}(e^{-\lambda u}\norm{p_0})-c_{d,\alpha}\right\}\,du
 &=\frac d\lambda\int_0^1r^{d-1}
     \int_0^r\frac{g_{d,\alpha}(s)-c_{d,\alpha}}s\,ds\,dr\\
 &=\frac1\lambda\int_0^1\frac{1-s^d}s
       \left(g_{d,\alpha}(s)-c_{d,\alpha}\right)\,ds
 =\frac{H_{d,\alpha}}{2\lambda}.
\end{align*}
Combining the two contributions proves Eq.\eqref{eq:tight-square-root-constant}.
Minimizing
\(\alpha s_d\lambda/4+H_{d,\alpha}/(2\lambda)\) proves
Eq.\eqref{eq:optimal-square-root-constant}.  If \(\alpha=2\), then
\(g_{d,2}(r)=s_d+r^2\).  Direct substitution in
Eq.\eqref{eq:H-constant} gives \(H_{d,2}=s_d\).
Substituting Eq.\eqref{eq:closed-H-bounds} into
Eq.\eqref{eq:optimal-square-root-constant} proves
Eq.\eqref{eq:closed-lambda-bounds}.  When \(d+\alpha>2\), substitution of
Eq.\eqref{eq:power-H-bounds} proves Eq.\eqref{eq:power-lambda-bounds}.  For
\((d,\alpha)=(1,1)\), the values follow from \(H_{1,1}=1/6\) and \(s_1=1/3\).
The bounds for the minimum of \(K_{d,\alpha}\) follow from
\(K_{d,\alpha}(\lambda^*_{d,\alpha})=\alpha s_d\lambda^*_{d,\alpha}/2\).

All estimates above are uniform when \(\lambda\) ranges over a compact subset of
\((0,\infty)\).  In the initial-state term, uniformity follows first on a
bounded \(u\)-interval from uniform Riemann-sum convergence and then on its
tail from the lower endpoint of the compact set and the quadratic envelope.
This proves the stated local uniformity.
\end{proof}

\begin{theorem}
\label{thm:global-finite-optimizer}
Under the hypotheses of Theorem~\ref{thm:square-root-constant}, let
\[
 \gamma_N^{\mathrm{opt}}\in
 \operatorname*{arg\,min}_{0\leq\gamma\leq1}
 \E L_{\alpha,N}^{(\gamma)}.
\]
Then
\begin{equation}
 \sqrt N\bigl(1-\gamma_N^{\mathrm{opt}}\bigr)
 \longrightarrow \lambda^*_{d,\alpha},
 \qquad
 \frac{\E L_{\alpha,N}^{(\gamma_N^{\mathrm{opt}})}
       -Nc_{d,\alpha}}{\sqrt N}
 \longrightarrow K_{d,\alpha}(\lambda^*_{d,\alpha}).
 \label{eq:global-finite-optimizer}
\end{equation}
\end{theorem}

\begin{proof}
Continuity in \(\gamma\) gives a minimizer.  Comparing it with
\(1-\lambda^*_{d,\alpha}N^{-1/2}\), and using
Theorem~\ref{thm:finite-envelope}, shows that
the accumulated state second moment at \(\gamma_N^{\mathrm{opt}}\) is
\(O(\sqrt N)\).
Put \(\beta_N=1-\gamma_N^{\mathrm{opt}}\).  The first term in
Eq.\eqref{eq:finite-state-sum} is at least \(N\beta_N/2\), so
\(\beta_N\sqrt N\) is bounded above.  If
\(\beta_N=0\), then the sum in Eq.\eqref{eq:finite-state-sum} equals \(Ns_d\),
so this endpoint cannot minimize
for large \(N\).  For \(0<\beta_N\leq1/2\),
 \[
  1-(1-\beta_N)^{2N}
  \geq1-e^{-2N\beta_N}
  \geq(1-e^{-1})\min\{2N\beta_N,1\}.
 \]
 The preceding upper bound gives \(\beta_N=O(N^{-1/2})\).  Consequently,
 \(\gamma_N^{\mathrm{opt}}\geq1/2\) for all sufficiently large \(N\).  For
 those \(N\), the contribution of the second term
 in Eq.\eqref{eq:finite-state-sum} is at least
 \[
  s_d\frac{1-e^{-1}}3
  \min\{2N,\beta_N^{-1}\}
  \geq s_d\frac{1-e^{-1}}3
  \min\{N,\beta_N^{-1}\}.
 \]
 If \(\beta_N\sqrt N\) tended to zero
along a subsequence, this lower bound would be larger than order \(\sqrt N\),
contradicting the preceding estimate.  Thus
\(\beta_N\sqrt N\) also stays bounded away from zero.

Every subsequence therefore has a further subsequence on which
\(\beta_N\sqrt N\to\lambda\in(0,\infty)\).  Local uniformity in
Theorem~\ref{thm:square-root-constant} and comparison with each fixed
competitor give \(K_{d,\alpha}(\lambda)\leq K_{d,\alpha}(\mu)\) for every
\(\mu>0\).  The unique minimizer of \(K_{d,\alpha}\) is
\(\lambda^*_{d,\alpha}\), so every subsequential limit equals this value.  The
cost convergence follows from the same local uniformity.
\end{proof}

\begin{theorem}
\label{thm:finite-quadratic}
For \(N\geq1\), with \(x_0=p_0\),
\begin{equation}
 \E L_{2,N}^{(\gamma)}
 =\frac{2s_d}{1+\gamma}
 \left[
 N+\gamma\frac{1-\gamma^{2N}}{1-\gamma^2}
 \right].
 \label{eq:finite-quadratic}
\end{equation}
At \(\gamma=1\), Eq.\eqref{eq:finite-quadratic} denotes its continuous
extension.  Both endpoints
satisfy
\[
 \E L_{2,N}^{(0)}=\E L_{2,N}^{(1)}=2Ns_d.
\]
For \(N\geq2\), the expression has a unique minimizer
\(\gamma_N^*\in(0,1)\).  If
\[
 P_N(\gamma)=\sum_{j=0}^{N-1}\gamma^{2j},
\]
then \(\gamma_N^*\) is the unique solution of
\begin{equation}
 P_N(\gamma)+\gamma(1+\gamma)P_N'(\gamma)=N.
 \label{eq:optimal-gamma}
\end{equation}
Moreover,
\begin{align}
 1-\gamma_N^*
 &=N^{-1/2}-\frac12N^{-1}+O(N^{-3/2}),
 \label{eq:optimal-gamma-asymptotic}\\
 \E L_{2,N}^{(\gamma_N^*)}
 &=s_dN+s_d\sqrt N+O(1).
 \label{eq:optimal-cost-asymptotic}
\end{align}
\end{theorem}

\begin{proof}
The point \(p_{t+1}\) is independent of \(x_t\), so
\(\E\norm{p_{t+1}-x_t}^2=s_d+\E\norm{x_t}^2\).
Summing Eq.\eqref{eq:finite-state-second} for \(t=0,\ldots,N-1\) proves
Eq.\eqref{eq:finite-quadratic}.  Direct substitution and passage to the limit at
\(\gamma=1\) give the endpoint values.

Differentiating Eq.\eqref{eq:finite-quadratic} shows that its derivative has the
sign of
\[
 P_N(\gamma)+\gamma(1+\gamma)P_N'(\gamma)-N.
\]
 For \(N\geq2\), this expression is negative at zero and positive at one.  Its
 derivative is
\[
 2(1+\gamma)P_N'(\gamma)
 +\gamma(1+\gamma)P_N''(\gamma)>0
\]
on \((0,1)\), which proves existence and uniqueness.

Let \(\beta_N=1-\gamma_N^*\) and
\(q_N=(1-\beta_N)^{2N}\).  In these variables, the left-hand side of
Eq.\eqref{eq:optimal-gamma} is
\begin{equation}
 (1-q_N)\frac{2-3\beta_N+2\beta_N^2}
 {\beta_N^2(2-\beta_N)}
 -\frac{2Nq_N}{\beta_N}.
 \label{eq:optimal-beta}
\end{equation}
For \(\beta=cN^{-1/2}\), the ratio of
Eq.\eqref{eq:optimal-beta} to \(N\) tends to \(c^{-2}\).  Evaluating at any
fixed \(c<1\) and \(c>1\), and using the strict monotonicity already proved,
gives \(\beta_N\sqrt N\to1\).  Consequently,
\[
 q_N=\exp\bigl(-(2+o(1))\sqrt N\bigr),
\]
so every term containing \(q_N\) in the optimality equation is \(o(1)\).
It follows that
\[
 \frac{2-3\beta_N+2\beta_N^2}
 {\beta_N^2(2-\beta_N)}=N+o(1).
\]
The rational function has the expansion
\[
 \frac{2-3\beta+2\beta^2}{\beta^2(2-\beta)}
 =\frac1{\beta^2}-\frac1\beta+\frac12+O(\beta).
\]
Multiplication by \(\beta_N^2\) gives
\[
 N\beta_N^2=1-\beta_N+O(\beta_N^2).
\]
Taking the positive square root gives
\[
 \sqrt N\,\beta_N
 =1-\frac{\beta_N}{2}+O(\beta_N^2).
\]
This first gives \(\beta_N=N^{-1/2}+O(N^{-1})\).  Substitution back gives
\[
 \beta_N=N^{-1/2}-\frac12N^{-1}+O(N^{-3/2}),
\]
which proves Eq.\eqref{eq:optimal-gamma-asymptotic}.  Finally,
\[
 \frac{\E L_{2,N}^{(\gamma_N^*)}}{s_d}
 =\frac{2}{2-\beta_N}
 \left[N+
 \frac{(1-\beta_N)(1-(1-\beta_N)^{2N})}
 {\beta_N(2-\beta_N)}\right]
 =N+\sqrt N+O(1),
\]
which proves Eq.\eqref{eq:optimal-cost-asymptotic}.
\end{proof}

The optimum has a simple interpretation.  The accumulated stationary excess
above the center-star cost grows on the scale \(N(1-\gamma)\), while
forgetting the random initial center takes scale \((1-\gamma)^{-1}\).  Their
balance selects \(1-\gamma\asymp N^{-1/2}\).
Theorem~\ref{thm:square-root-constant} determines the
balance constant for general \(\alpha\), and
Theorem~\ref{thm:finite-quadratic} supplies the complete formula when
\(\alpha=2\).

\subsection{A time-varying comparison}

The preceding optimization keeps the same value of \(\gamma\) throughout the
construction.  A useful comparison is obtained by attaching the next point to
the mean of all points already observed.  Let
\begin{equation}
 \bar x_t=\frac1{t+1}\sum_{j=0}^t p_j,
 \qquad
 \bar L_{\alpha,N}
 =\sum_{t=0}^{N-1}\norm{p_{t+1}-\bar x_t}^\alpha.
 \label{eq:running-average}
\end{equation}

\begin{proposition}
\label{prop:running-average}
For every \(d\geq1\), \(\alpha>0\), and \(N\geq1\),
\begin{equation}
 Nc_{d,\alpha}+k^-_{d,\alpha}s_d\sum_{j=1}^N\frac1j
 \leq \E\bar L_{\alpha,N}
 \leq Nc_{d,\alpha}+k^+_{d,\alpha}s_d\sum_{j=1}^N\frac1j.
 \label{eq:running-average-envelope}
\end{equation}
For \(\alpha=2\),
\begin{equation}
 \E\bar L_{2,N}=Ns_d+s_d\sum_{j=1}^N\frac1j.
 \label{eq:running-average-quadratic}
\end{equation}
\end{proposition}

\begin{proof}
The points are centered and independent, so
\[
 \E\norm{\bar x_t}^2=\frac{s_d}{t+1}.
\]
The point \(p_{t+1}\) is independent of \(\bar x_t\).  Apply
Eq.\eqref{eq:best-quadratic-envelope}, take expectations, and sum over
\(t=0,\ldots,N-1\).  For \(\alpha=2\), the identity
\(g_{d,2}(r)=s_d+r^2\) gives Eq.\eqref{eq:running-average-quadratic}.
\end{proof}

For arbitrary points in \(\B\), let
\[
 \overline M_{\alpha,N}^{(d)}
 =\sup_{p_0,\ldots,p_N\in\B}\bar L_{\alpha,N}.
\]

\begin{proposition}
\label{prop:running-average-adversarial}
For every input sequence and \(N\geq1\),
\begin{equation}
 \sum_{t=1}^N\norm{p_t-\bar x_{t-1}}^2
 \leq N+\sum_{t=1}^N\frac1t+2.
 \label{eq:running-average-adversarial-bound}
\end{equation}
Consequently, for every \(d\geq1\) and \(0<\alpha\leq2\),
\begin{equation}
 \lim_{N\to\infty}\frac1N\overline M_{\alpha,N}^{(d)}=1.
 \label{eq:running-average-adversarial-value}
\end{equation}
For \(\alpha>2\),
\begin{equation}
 \limsup_{N\to\infty}\frac1N\overline M_{\alpha,N}^{(d)}
 \geq 2^\alpha\alpha^{-\alpha/(\alpha-1)}>1.
 \label{eq:running-average-block-lower}
\end{equation}
\end{proposition}

\begin{proof}
Let \(m_t=\bar x_t\).  The update
\(m_t=(t m_{t-1}+p_t)/(t+1)\) gives the exact identity
\[
 t\norm{m_{t-1}}^2+\norm{p_t}^2-(t+1)\norm{m_t}^2
 =\frac{t}{t+1}\norm{p_t-m_{t-1}}^2.
\]
After multiplication by \((t+1)/t\) and summation, the coefficients of
\(\norm{m_t}^2\) for \(1\leq t<N\) equal \(-1/t\), and the final coefficient
is negative.  Since all input points have norm at most one,
\[
 \sum_{t=1}^N\norm{p_t-m_{t-1}}^2
 \leq 2\norm{p_0}^2+
 \sum_{t=1}^N\left(1+\frac1t\right),
\]
which proves Eq.\eqref{eq:running-average-adversarial-bound}.  Power-mean
monotonicity gives
\[
 \frac1N\sum_{t=1}^N\norm{p_t-m_{t-1}}^\alpha
 \leq\left(1+\frac{2+\sum_{t=1}^N1/t}{N}\right)^{\alpha/2}
\]
for \(0<\alpha\leq2\).  Alternating the endpoints of a diameter makes the
left side converge to one, proving Eq.\eqref{eq:running-average-adversarial-value}.

For the last assertion, choose a unit vector \(u\), put
\(m=\lfloor\theta N\rfloor\), take the first \(m\) points equal to \(u\), and
take the remaining points equal to \(-u\).  During the second block, the
insertion at time \(t\) has length \(2m/t\).  Hence its limiting mean cost is
\[
 \frac{2^\alpha\theta(1-\theta^{\alpha-1})}{\alpha-1}.
\]
The maximum occurs at \(\theta=\alpha^{-1/(\alpha-1)}\) and equals the right
side of Eq.\eqref{eq:running-average-block-lower}.
\end{proof}

The update
\[
 \bar x_t=\frac{t}{t+1}\bar x_{t-1}+\frac1{t+1}p_t
\]
 stores one point and the insertion index, and its coefficients vary with
 \(t\).  Proposition~\ref{prop:running-average} gives an excess of order
 \(\log N\) over \(Nc_{d,\alpha}\).  Under the hypotheses of
 Theorem~\ref{thm:global-finite-optimizer}, the best rule with one constant value
 of \(\gamma\) has an excess of order \(\sqrt N\).  This comparison identifies the
scope of the finite-sample optimization: constant exponential weights give a
stationary one-point recursion with a fixed coefficient profile, while the
 running mean assigns equal weight to all \(t+1\) observed points.  Their
 insertion times are, on average, \(t/2\) steps before the most recent point.
 Proposition~\ref{prop:running-average-adversarial} shows that the running mean
 also has asymptotic adversarial value one for \(0<\alpha\leq2\), while a
 two-block input gives a value strictly larger than one for \(\alpha>2\).

Appendix~\ref{app:running-mean} records a sharper stochastic comparison:
among deterministic convex averages of the observed points, the running mean
minimizes the next insertion length in stochastic order.

For ordinary Euclidean length, let \(\mu_d\) be the mean distance between two
independent uniform points in \(\B\).  The input-order path and the star
centered at \(p_0\) satisfy
\begin{equation}
 \E L_{1,N}^{(0)}=\E L_{1,N}^{(1)}=N\mu_d.
 \label{eq:length-endpoints}
\end{equation}
Since \(g_{d,1}(r)>c_{d,1}\) for \(r>0\), we have
\(\mu_d=\E g_{d,1}(\norm{p_0})>c_{d,1}\).  The optimized exponential rule
therefore has a strictly smaller leading constant than either endpoint, and
Theorem~\ref{thm:square-root-constant} identifies its optimal square-root
correction.

\subsection{Stationary law and moments}
\label{sec:stationary-main}

We use the peakedness order for symmetric random vectors: \(y_1\) is more
peaked than \(y_2\) if
\[
 \Pr(y_1\in K)\geq\Pr(y_2\in K)
\]
for every compact convex set \(K\subset\R^d\) symmetric about the origin.  The
majorization theorem needed below is the following.

\begin{lemma}
\label{lem:olkin-tong}
Let \(z_0,\ldots,z_t\) be independent random vectors in \(\R^d\) with a common
symmetric log-concave density.  Let \(a,b\in[0,\infty)^{t+1}\) have the same
sum.  If \(a\) majorizes \(b\), then \(\sum_j b_jz_j\) is more peaked than
\(\sum_j a_jz_j\).
\end{lemma}

This is the multivariate theorem of Olkin and Tong~\cite{OlkinTong1988}.
Proposition~8.6 of Saumard and Wellner~\cite{SaumardWellner2014} states the
same result with this ordering convention.  It extends the scalar
convex-combination result of Proschan~\cite{Proschan1965}.

\begin{proposition}
\label{prop:stationary-series}
The series
\begin{equation}
 x_\gamma=(1-\gamma)\sum_{j=0}^\infty\gamma^j p_j
 \label{eq:stationary-series}
\end{equation}
converges absolutely almost surely.  Its distribution is invariant under
orthogonal transformations, has support \(\B\), and is stationary for
Eq.\eqref{eq:recursion}.  It is the unique stationary distribution supported
on \(\B\) when \(0\leq\gamma<1\).
\end{proposition}

\begin{proof}
Since \(\norm{p_j}\leq1\),
\[
 (1-\gamma)\sum_{j=0}^\infty\gamma^j\norm{p_j}\leq1,
\]
 which proves absolute convergence.  If \(Q\) is an orthogonal
 \(d\times d\) matrix, then
\((Qp_j)_{j\geq0}\) has the same joint distribution as
\((p_j)_{j\geq0}\), which proves rotational invariance.
 The triangle inequality gives \(\norm{x_\gamma}\leq1\).

 If \(\gamma=0\), then \(x_\gamma=p_0\), so the support assertion is immediate.
 Assume \(0<\gamma<1\) for the following support argument.
 Let \(z\) be in the interior of \(\B\), and let \(\varepsilon>0\).
Choose \(N\) so that \(2\gamma^{N+1}<\varepsilon/2\).  With positive
probability, each of \(p_0,\ldots,p_N\) belongs to a sufficiently small ball
about \(z\).  On this event, the corresponding finite weighted sum is within
\(\varepsilon/2\) of \((1-\gamma^{N+1})z\), and the remaining tail has norm at
most \(\gamma^{N+1}\).  Hence \(x_\gamma\) is within \(\varepsilon\) of \(z\)
with positive probability.  Taking closures gives support \(\B\).

Finally, let \(p\) be an additional uniform point, independent of
\((p_j)_{j\geq0}\).  The sequence \(p,p_0,p_1,\ldots\) is again independent
and uniformly distributed in \(\B\), and
\[
 \gamma x_\gamma+(1-\gamma)p
 =(1-\gamma)\left(p+\sum_{j=1}^\infty\gamma^jp_{j-1}\right).
\]
Thus the updated state has the distribution defined by
Eq.\eqref{eq:stationary-series}, which proves stationarity.

For uniqueness, let two stationary states supported on \(\B\) be coupled as
\(x_0\) and \(y_0\) and driven by the same independent input sequence.
Subtracting the two recursions and iterating gives
\[
 \norm{x_t-y_t}=\gamma^t\norm{x_0-y_0}\leq2\gamma^t.
\]
For every bounded Lipschitz function \(f\), stationarity therefore gives
\[
 |\E f(x_0)-\E f(y_0)|
 =|\E f(x_t)-\E f(y_t)|
 \leq2\operatorname{Lip}(f)\gamma^t.
\]
Letting \(t\) tend to infinity shows that the two laws agree.
\end{proof}

\begin{theorem}
\label{thm:stochastic-monotonicity}
Let \(p\) be a uniform point in \(\B\), independent of the stationary state,
and put
\[
 Z_\gamma=\norm{p-x_\gamma}.
\]
If \(0\leq\gamma_1<\gamma_2<1\), then
\begin{equation}
 \Pr(Z_{\gamma_2}>r)\leq\Pr(Z_{\gamma_1}>r)
 \quad\hbox{for every }r\geq0.
 \label{eq:stochastic-monotonicity}
\end{equation}
Consequently, for every \(\alpha>0\),
\begin{equation}
 \E Z_{\gamma_2}^\alpha
 \leq \E Z_{\gamma_1}^\alpha,
 \qquad
 \lim_{\gamma\to1^-}\E Z_\gamma^\alpha
 =c_{d,\alpha}.
 \label{eq:all-moment-monotonicity}
\end{equation}
\end{theorem}

\begin{proof}
For \(n\geq0\), consider the normalized coefficient vector
\[
 w^{(n)}(\gamma)
 =\left(\frac{(1-\gamma)\gamma^j}{1-\gamma^{n+1}}
   \right)_{0\leq j\leq n}.
\]
Its first \(m+1\) entries have sum
\[
 \frac{1-\gamma^{m+1}}{1-\gamma^{n+1}},
 \qquad 0\leq m<n,
\]
which is strictly decreasing in \(\gamma\).  Indeed, if
\(S_k(\gamma)=\sum_{j=0}^k\gamma^j\), then, for \(0<\gamma<1\),
\[
 \frac{d}{d\gamma}\frac{S_m(\gamma)}{S_n(\gamma)}
 =\frac{\displaystyle
   \sum_{i=0}^m\sum_{j=m+1}^n(i-j)\gamma^{i+j-1}}
  {S_n(\gamma)^2}<0.
\]
Continuity supplies the comparison at \(\gamma=0\).  Hence
\(w^{(n)}(\gamma_1)\) majorizes \(w^{(n)}(\gamma_2)\).  The uniform density on
\(\B\) is symmetric and log-concave in the extended-value sense.
Lemma~\ref{lem:olkin-tong} therefore shows that the sum with coefficients
\(w^{(n)}(\gamma_2)\) is more peaked than the sum with coefficients
\(w^{(n)}(\gamma_1)\).

The normalizing factors also preserve the required comparison.  Put
\[
 y_{i,n}=\sum_{j=0}^n w_j^{(n)}(\gamma_i)p_j,
 \qquad \lambda_{i,n}=1-\gamma_i^{n+1},
 \qquad i\in\{1,2\}.
\]
Then \(0<\lambda_{2,n}<\lambda_{1,n}\).  If \(K\) is compact, convex, and
symmetric about the origin, then
\(K/\lambda_{1,n}\subset K/\lambda_{2,n}\), and hence
\begin{align*}
 \Pr(\lambda_{2,n}y_{2,n}\in K)
 &=\Pr(y_{2,n}\in K/\lambda_{2,n})\\
 &\geq\Pr(y_{1,n}\in K/\lambda_{2,n})\\
 &\geq\Pr(y_{1,n}\in K/\lambda_{1,n})
 =\Pr(\lambda_{1,n}y_{1,n}\in K).
\end{align*}
Thus the unnormalized truncated sum for \(\gamma_2\) is more peaked than the
one for \(\gamma_1\).  Letting \(n\to\infty\) gives \(x_\gamma\) almost surely.
Each finite sum has an absolutely continuous law.  The limiting law is also absolutely continuous,
because its series contains the independent term \((1-\gamma)p_0\).
Boundaries of convex sets with nonempty interior have Lebesgue measure zero;
lower-dimensional convex sets have probability zero under all these laws.
Thus every compact symmetric convex set needed in the definition is a
continuity set.  The peakedness inequalities
therefore pass to the limit, and \(x_{\gamma_2}\) is more peaked than
\(x_{\gamma_1}\).

In particular,
\[
 \Pr(\norm{x_{\gamma_2}}\leq r)
 \geq\Pr(\norm{x_{\gamma_1}}\leq r)
 \quad(r\geq0).
\]
For fixed \(r\geq0\), the function
\[
 z\longmapsto\Pr(\norm{p-z}\leq r)
 =\frac{\operatorname{vol}(\B\cap B(z,r))}{\operatorname{vol}(\B)}
\]
 depends only on \(\norm z\) and is nonincreasing in that radius by
 Lemma~\ref{lem:translated-ball-overlap}.  Its superlevel sets are balls
 centered at the origin.  Every open superlevel ball is an increasing union of
 compact centered balls, so the peakedness comparison extends to it by
 continuity from below.  The layer-cake representation and the peakedness
 order therefore give
\begin{align*}
 \Pr(Z_{\gamma_2}\leq r)
 &=\E\frac{\operatorname{vol}(\B\cap B(x_{\gamma_2},r))}
                 {\operatorname{vol}(\B)}\\
 &\geq\E\frac{\operatorname{vol}(\B\cap B(x_{\gamma_1},r))}
                 {\operatorname{vol}(\B)}
 =\Pr(Z_{\gamma_1}\leq r).
\end{align*}
This is equivalent to Eq.\eqref{eq:stochastic-monotonicity}.  The moment inequality is
immediate.  Finally, Eq.\eqref{eq:x-second} below gives
\(x_\gamma\to O\) in probability as \(\gamma\to1^-\), and hence
\(Z_\gamma\to\norm p\) in probability.  The variables \(Z_\gamma^\alpha\)
are bounded by \(2^\alpha\) and therefore uniformly integrable.  Their
expectations consequently converge to
\(\E\norm p^\alpha=c_{d,\alpha}\), proving the limit in
Eq.\eqref{eq:all-moment-monotonicity}.
\end{proof}

\begin{proposition}
\label{prop:stationary-moment-bounds}
The stationary state satisfies
\begin{equation}
 \E\norm{x_\gamma}^2
 =s_d\frac{1-\gamma}{1+\gamma}.
 \label{eq:x-second}
\end{equation}
If \(p\) is an additional independent uniform point and
\(Z_\gamma=\norm{p-x_\gamma}\), then, for every \(\alpha>0\),
\begin{equation}
 c_{d,\alpha}
 +k^-_{d,\alpha}s_d\frac{1-\gamma}{1+\gamma}
 \leq\E Z_\gamma^\alpha\leq
 c_{d,\alpha}
 +k^+_{d,\alpha}s_d\frac{1-\gamma}{1+\gamma}.
 \label{eq:stationary-moment-envelope}
\end{equation}
For \(\alpha=2\), this becomes
\begin{equation}
 \E Z_\gamma^2
 =\frac{2d}{(d+2)(1+\gamma)}.
 \label{eq:exact-quadratic-stationary}
\end{equation}
\end{proposition}

\begin{proof}
Apply the finite-sum identity for independent centered vectors to the first
\(n+1\) terms of Eq.\eqref{eq:stationary-series}.  Since the remaining tail has
norm at most \(\gamma^{n+1}\), passage in \(L^2\) gives
\[
 \E\norm{x_\gamma}^2
 =s_d(1-\gamma)^2\sum_{j\geq0}\gamma^{2j}
 =s_d\frac{1-\gamma}{1+\gamma}.
\]
Conditional on \(x_\gamma\), rotational invariance and
Eq.\eqref{eq:radial-mean-cost} give
\[
 \E\bigl[Z_\gamma^\alpha\mid x_\gamma\bigr]
 =g_{d,\alpha}(\norm{x_\gamma}).
\]
Apply Eq.\eqref{eq:best-quadratic-envelope}, take expectations, and use
Eq.\eqref{eq:x-second}.  This proves
Eq.\eqref{eq:stationary-moment-envelope}.  When \(\alpha=2\), both quadratic
envelope constants equal one, which gives
Eq.\eqref{eq:exact-quadratic-stationary}.
\end{proof}

Consequently, as \(\gamma\) tends to one, \(Z_\gamma\) converges in
distribution to the radius \(\rho=\norm p\), whose distribution function is
\(\Pr(\rho\leq r)=r^d\) on \([0,1]\).  The difference between the expected
insertion cost and \(d/(d+\alpha)\), the cost of the center-\(O\) star, is
bounded above and below by explicit multiples of
\((1-\gamma)/(1+\gamma)\).  Appendices~\ref{app:stationary-transforms} and
\ref{app:stationary-limits} give exact transforms, fourth moments, a Gaussian
limit for the rescaled state, and a Wasserstein bound for the insertion length.

Figure~\ref{fig:radial-laws} displays the radial law for dimensions from one to
ten and its concentration near the boundary.  The rescaling
\(Y_d=d(1-\rho)\) gives the nondegenerate limit
\[
 \Pr(Y_d\leq y)
 =1-\left(1-\frac{y}{d}\right)^d
 \longrightarrow 1-e^{-y}
\]
for \(0\leq y\leq d\), and hence for every fixed \(y\geq0\) as \(d\) tends
to infinity.

\begin{figure}[!ht]
\centering
\includegraphics[width=0.98\textwidth]{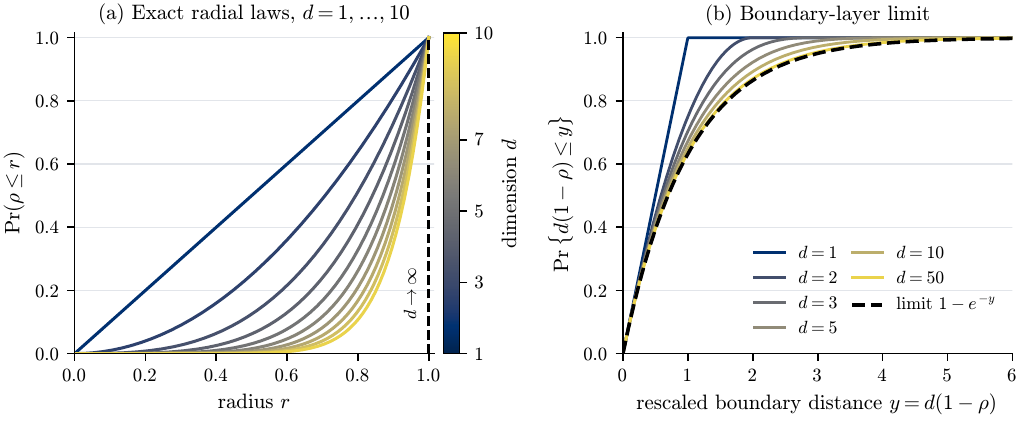}
\caption{Distribution of the distance from a uniform point in the unit ball to
its center.  Panel (a) superposes the exact distribution functions \(r^d\) for
\(d=1,\ldots,10\); the dashed vertical line is the degenerate limit at
\(r=1\).  Panel (b) shows the distribution functions of \(d(1-\rho)\) and
their limiting exponential distribution.}
\label{fig:radial-laws}
\end{figure}

\section{Adversarial sequences}
\label{sec:adversarial}

For \(0\leq\gamma<1\), define
\[
 M_{n,\alpha}^{(d)}(\gamma;x_0)
 =\sup_{p_1,\ldots,p_n\in\B}
   \sum_{i=1}^n\norm{p_i-x_{i-1}}^\alpha
\]
and
\begin{equation}
 \A_{d,\alpha}(\gamma)
 =\limsup_{n\to\infty}\frac1n
 M_{n,\alpha}^{(d)}(\gamma;x_0).
 \label{eq:adversarial-value}
\end{equation}
For each \(n\), the supremum over the first \(n\) input points is taken before
the limsup as \(n\) tends to infinity.  Eq.\eqref{eq:weighted-adversarial-value}
uses the same order of operations.

The following two estimates control the effect of the initial
state in this section.

\begin{lemma}
\label{lem:contraction}
Let \((x_i)\) and \((y_i)\) satisfy Eq.\eqref{eq:recursion} for the same input
sequence, with respective initial states \(x_0,y_0\in\B\).  Then
\[
 x_i-y_i=\gamma^i(x_0-y_0)
 \quad\hbox{and}\quad
 \norm{x_i-y_i}=\gamma^i\norm{x_0-y_0}.
\]
\end{lemma}

\begin{proof}
Subtracting the two recursions gives
\(x_i-y_i=\gamma(x_{i-1}-y_{i-1})\).  Iteration proves both identities.
\end{proof}

\begin{lemma}
\label{lem:power-difference}
For \(a,b\in[0,2]\),
\[
 |a^\alpha-b^\alpha|
 \leq
 \begin{cases}
  |a-b|^\alpha, & 0<\alpha\leq1,\\
  \alpha2^{\alpha-1}|a-b|, & \alpha\geq1.
 \end{cases}
\]
\end{lemma}

\begin{proof}
For \(0<\alpha\leq1\), subadditivity of \(t^\alpha\) gives
\(a^\alpha\leq b^\alpha+|a-b|^\alpha\), and the same argument with \(a,b\)
interchanged proves the claim.  For \(\alpha\geq1\), the result follows from
the mean value theorem on \([0,2]\).
\end{proof}

\begin{proposition}
\label{prop:initial-state}
For \(0\leq\gamma<1\), the value in Eq.\eqref{eq:adversarial-value} is independent of
\(x_0\in\B\).
\end{proposition}

\begin{proof}
Drive two states by the same input sequence.  Lemma~\ref{lem:contraction}
gives their distance at every step, and Lemma~\ref{lem:power-difference}
bounds the difference between the corresponding insertion costs by
\[
 \begin{cases}
  \gamma^{\alpha(i-1)}\norm{x_0-y_0}^\alpha,
       &0<\alpha\leq1,\\
  \alpha2^{\alpha-1}\gamma^{i-1}\norm{x_0-y_0},
       &\alpha\geq1.
 \end{cases}
\]
In either case, summation gives a constant \(D\), independent of \(n\) and of
the input sequence, such that the two accumulated costs differ by at most
\(D\).  Taking the supremum over the same set of length-\(n\) input sequences
in both directions gives
\[
 |M_{n,\alpha}^{(d)}(\gamma;x_0)
   -M_{n,\alpha}^{(d)}(\gamma;y_0)|\leq D.
\]
After division by \(n\), the bound tends to zero.
\end{proof}

\begin{proposition}
\label{prop:alternation}
For \(0\leq\gamma<1\) and every \(\alpha>0\),
\[
 \A_{d,\alpha}(\gamma)
 \geq\left(\frac2{1+\gamma}\right)^\alpha.
\]
\end{proposition}

\begin{proof}
Fix a unit vector \(u\) and take \(p_i=(-1)^{i-1}u\).  On the attracting
two-periodic orbit, the state before the insertion of \(u\) is
\(-((1-\gamma)/(1+\gamma))u\).  After the insertion it is
\(((1-\gamma)/(1+\gamma))u\).  Thus every insertion on the
periodic orbit has length
\[
 1+\frac{1-\gamma}{1+\gamma}=\frac2{1+\gamma}.
\]
Lemma~\ref{lem:contraction}, followed by Lemma~\ref{lem:power-difference},
shows that the total difference between the transient insertion costs and
those on the attracting orbit is bounded.  Its contribution to the mean cost
therefore vanishes.
\end{proof}

Set
\[
 A_\gamma=\frac2{1+\gamma}.
\]

The next potential inequality proves the matching upper bound for
\(\alpha=2\).  It also provides the endpoint used for \(0<\alpha<2\).

\begin{lemma}
\label{lem:quadratic}
Let \(0\leq\gamma<1\).  If \(y=\gamma x+(1-\gamma)p\), with
\(x,p\in\B\), then
\begin{equation}
 \norm{p-x}^2
 +\frac2{1-\gamma^2}\bigl(\norm y^2-\norm x^2\bigr)
 \leq A_\gamma^2.
 \label{eq:quadratic-potential}
\end{equation}
\end{lemma}

\begin{proof}
Let \(r=\norm x\), \(s=\norm p\), and \(c=\langle x,p\rangle\).
The left side of Eq.\eqref{eq:quadratic-potential} is
\[
 -r^2+\frac{3-\gamma}{1+\gamma}s^2
-\frac{2(1-\gamma)}{1+\gamma}c.
\]
For fixed \(r\) and \(s\), the expression is largest when \(p\) points in the
direction opposite to \(x\), so \(c=-rs\).  The resulting expression is
increasing in \(s\in[0,1]\), and its maximum therefore occurs at \(s=1\).  It
becomes
\[
 -r^2+\frac{3-\gamma}{1+\gamma}
 +\frac{2(1-\gamma)}{1+\gamma}r.
\]
This concave quadratic in \(r\in[0,1]\) is maximized at
\(r=(1-\gamma)/(1+\gamma)\).  Substitution gives
\(4/(1+\gamma)^2=A_\gamma^2\).
\end{proof}

\begin{theorem}
\label{thm:through-two}
For \(0\leq\gamma<1\) and \(0<\alpha\leq2\),
\[
 \A_{d,\alpha}(\gamma)=A_\gamma^\alpha.
\]
\end{theorem}

\begin{proof}
By concavity of \(z^{\alpha/2}\), for every \(z\geq0\),
\[
 z^{\alpha/2}\leq A_\gamma^\alpha
 +\frac\alpha2 A_\gamma^{\alpha-2}(z-A_\gamma^2).
\]
Apply this inequality to \(z=\norm{p-x}^2\), and then use
Lemma~\ref{lem:quadratic}.  We obtain
\[
 \norm{p-x}^\alpha
 \leq A_\gamma^\alpha
 +\frac{\alpha A_\gamma^{\alpha-2}}{1-\gamma^2}
   \bigl(\norm x^2-\norm y^2\bigr).
\]
Summing over the first \(n\) insertions gives
\[
 \sum_{i=1}^n\norm{p_i-x_{i-1}}^\alpha
 \leq nA_\gamma^\alpha
 +\frac{\alpha A_\gamma^{\alpha-2}}{1-\gamma^2}
   \bigl(\norm{x_0}^2-\norm{x_n}^2\bigr).
\]
The final term is bounded independently of \(n\).  Division by \(n\) gives
\(\A_{d,\alpha}(\gamma)\leq A_\gamma^\alpha\).
Proposition~\ref{prop:alternation} gives the reverse inequality.
\end{proof}

\begin{lemma}
\label{lem:cubic}
Let \(0\leq\gamma<1\), and put
\[
 h_\gamma(x)
 =\frac{11+2\gamma-\gamma^2}
       {2(1-\gamma)(1+\gamma)^2}\norm x^2
  +\frac1{4(1-\gamma)}\norm x^4.
\]
For \(x,p\in\B\) and \(y=\gamma x+(1-\gamma)p\),
\begin{equation}
 \norm{p-x}^3+h_\gamma(y)-h_\gamma(x)
 \leq A_\gamma^3.
 \label{eq:cubic-potential}
\end{equation}
\end{lemma}

\begin{proof}
Fix \(x\).  The function
\[
 p\longmapsto \norm{p-x}^3+h_\gamma(\gamma x+(1-\gamma)p)
\]
is convex.  Its maximum over \(\B\) is therefore attained on the unit sphere.
If \(x=O\), rotational symmetry allows any diameter through \(p\), and the
one-dimensional argument below applies with \(x=0\).  Assume henceforth that
\(r=\norm x>0\), and set \(t=\langle x/r,p\rangle\).
The term
\[
 (1+r^2-2rt)^{3/2}
\]
is convex in \(t\in[-1,1]\).  Moreover,
\(\norm{\gamma x+(1-\gamma)p}^2\) is affine in \(t\), and the quadratic
coefficient produced by its square is nonnegative.  Thus the complete
expression is convex in \(t\), and its maximum is attained at \(t=-1\) or
\(t=1\).

Identify the chosen diameter with \([-1,1]\).  Both cases are covered by
taking \(p=1\) and \(x\in[-1,1]\).  Since
\(\lvert1-x\rvert^3=(1-x)^3\), direct expansion gives
\begin{align}
 &A_\gamma^3+h_\gamma(x)
 -h_\gamma(\gamma x+1-\gamma)-(1-x)^3 \notag\\
 &\quad=
 \frac{\bigl((1+\gamma)x+1-\gamma\bigr)^2}
 {4(1+\gamma)^3}Q_\gamma(x),
 \label{eq:cubic-factorization}
\end{align}
where
\begin{align*}
 Q_\gamma(x)={}&
 (\gamma^4+2\gamma^3+2\gamma^2+2\gamma+1)x^2\\
 &+(-2\gamma^4-4\gamma^3+4\gamma+2)x\\
 &+\gamma^4+2\gamma^3-2\gamma^2-6\gamma+5.
\end{align*}
The leading coefficient of \(Q_\gamma\) is positive, and its discriminant is
\[
 -16(1-\gamma)^2(1+\gamma)^2<0.
\]
 Hence \(Q_\gamma(x)>0\) for every real \(x\).  The right side of
 Eq.\eqref{eq:cubic-factorization} is nonnegative, which proves
 Eq.\eqref{eq:cubic-potential}.
 Equality occurs when \(p\) is antipodal to \(x\) and
 \(\norm x=(1-\gamma)/(1+\gamma)\), exactly the state on the attracting
 alternation orbit.
\end{proof}

\begin{theorem}
\label{thm:through-three}
For \(0\leq\gamma<1\) and \(0<\alpha\leq3\),
\begin{equation}
 \A_{d,\alpha}(\gamma)
 =\left(\frac2{1+\gamma}\right)^\alpha.
 \label{eq:main-value}
\end{equation}
\end{theorem}

\begin{proof}
Theorem~\ref{thm:through-two} covers \(0<\alpha\leq2\).
For \(\alpha=3\), summing Eq.\eqref{eq:cubic-potential} over the first \(n\)
insertions gives
\[
 \sum_{i=1}^n\norm{p_i-x_{i-1}}^3
 \leq nA_\gamma^3+h_\gamma(x_0)-h_\gamma(x_n).
\]
The potential is bounded on \(\B\), so division by \(n\) yields the desired
upper bound.  Proposition~\ref{prop:alternation} gives equality.

It remains to consider \(2<\alpha<3\).  Choose \(\theta\in(0,1)\) such that
\[
 \frac1\alpha=\frac{\theta}{2}+\frac{1-\theta}{3}.
\]
For the insertion lengths \(\ell_i=\norm{p_i-x_{i-1}}\), H\"older's inequality,
applied to the uniform probability measure on \(\{1,\ldots,n\}\), gives
\[
 \left(\frac1n\sum_{i=1}^n\ell_i^\alpha\right)^{1/\alpha}
 \leq
 \left(\frac1n\sum_{i=1}^n\ell_i^2\right)^{\theta/2}
 \left(\frac1n\sum_{i=1}^n\ell_i^3\right)^{(1-\theta)/3}.
\]
 Since both potentials are bounded on \(\B\), there are constants
 \(C_2,C_3\), independent of the input sequence and \(n\), such that
 \[
  \frac1n\sum_{i=1}^n\ell_i^2
  \leq A_\gamma^2+\frac{C_2}{n},
  \qquad
  \frac1n\sum_{i=1}^n\ell_i^3
  \leq A_\gamma^3+\frac{C_3}{n}.
 \]
 Hence the two factors in the preceding inequality are
 \(A_\gamma+o(1)\), uniformly over the input sequence.  Thus the limsup of the
left side is at most \(A_\gamma\).  Proposition~\ref{prop:alternation} again
gives equality.
\end{proof}

\begin{corollary}
\label{cor:same-N}
Assume \(0<\alpha\leq3\) and either \(d+\alpha>2\) or
\((d,\alpha)=(1,1)\).  Fix \(\lambda>0\) and, for all sufficiently large
\(N\), set \(\gamma_N=1-\lambda N^{-1/2}\).  Then
\begin{align}
 \E L_{\alpha,N}^{(\gamma_N)}
 &=Nc_{d,\alpha}+K_{d,\alpha}(\lambda)\sqrt N+o(\sqrt N),
 \label{eq:same-N-uniform}\\
 \frac1N M_{N,\alpha}^{(d)}(\gamma_N;x_0)
 &=1+O(N^{-1/2})
 \label{eq:same-N-adversarial}
\end{align}
uniformly for \(x_0\in\B\).  Thus the same value of \(\gamma_N\) gives the
sharp uniform-input correction and a worst-case mean cost converging to the
unit scale for the same value of \(N\).
The constants implicit in the \(O\)-term may depend on
\(d,\alpha,\lambda\); the stated uniformity is with respect to \(x_0\in\B\).
\end{corollary}

\begin{proof}
Eq.\eqref{eq:same-N-uniform} is
Theorem~\ref{thm:square-root-constant}.  Put
\(A_N=2/(1+\gamma_N)\).  If \(0<\alpha\leq2\), the finite inequality in the
proof of Theorem~\ref{thm:through-two} gives
\[
 M_{N,\alpha}^{(d)}(\gamma_N;x_0)
 \leq N A_N^\alpha
 +\frac{\alpha A_N^{\alpha-2}}{1-\gamma_N^2}.
\]
Here \(A_N=1+O(N^{-1/2})\) and
\((1-\gamma_N^2)^{-1}=O(\sqrt N)\).

For \(\alpha=3\), summing Eq.\eqref{eq:cubic-potential} gives
\[
 M_{N,3}^{(d)}(\gamma_N;x_0)
 \leq N A_N^3+\max_{x\in\B}h_{\gamma_N}(x).
\]
The two positive coefficients defining \(h_{\gamma_N}\) are
\(O((1-\gamma_N)^{-1})\), so the last maximum is \(O(\sqrt N)\).  For
\(2<\alpha<3\), apply the power-mean interpolation used in the proof of
Theorem~\ref{thm:through-three} to these finite quadratic and cubic estimates.
This proves the upper bound in Eq.\eqref{eq:same-N-adversarial}.

For the lower bound, use the alternating diameter sequence and compare the
state from \(x_0\) with its attracting two-periodic state.  Lemmas
\ref{lem:contraction} and \ref{lem:power-difference} bound the difference of
the two accumulated costs by
\[
 \frac{2^\alpha}{1-\gamma_N^\alpha}
 \quad (0<\alpha\leq1),
 \qquad
 \frac{\alpha2^\alpha}{1-\gamma_N}
 \quad (\alpha\geq1).
\]
Both quantities are \(O(\sqrt N)\), while the periodic state has cost
\(N A_N^\alpha\).  Hence the supremum is at least
\(N A_N^\alpha-O(\sqrt N)\).  Since \(A_N^\alpha=1+O(N^{-1/2})\), the matching
lower estimate follows.
\end{proof}

The calibration in Corollary~\ref{cor:same-N} assumes that the total number
\(N\) of insertions is known.  Epochs of lengths \(1,2,4,\ldots\) can define a
horizon-free rule, with the state restarted at epoch boundaries.  The restart
costs and the distributional correction of that rule require a separate
analysis and are not asserted here.

The restriction \(0<\alpha\leq3\) is substantive.  For every fixed
\(0<\gamma<1\), Appendix~\ref{app:periodic-blocks} gives periodic endpoint
blocks whose adversarial mean cost is strictly larger than that of alternation
when \(\alpha\) is sufficiently large.

\subsection{Comparison with general weighted averages}
\label{sec:weighted-adversarial}

We now use a fixed sequence of weights, the same at every insertion, to average
the previously inserted points.  Let
\begin{equation}
 w_j\geq0,\qquad
 \sum_{j\geq0}w_j=1,\qquad
 a=\sum_{j\geq0}jw_j<\infty.
 \label{eq:weighted-mean-lag}
\end{equation}
The coefficient \(w_j\) is assigned to \(p_{t-1-j}\), the point inserted \(j\)
steps before the most recently inserted point \(p_{t-1}\).  Let \(J\) be an
integer-valued random variable with \(\Pr(J=j)=w_j\).  The index \(J\) is the
lag of the selected earlier point and \(a=\E J\) is the mean lag.  Equivalently,
a point selected according to the weights lies, on average, \(a+1\) insertions
before the incoming point \(p_t\).  The constraint therefore fixes the
temporal centroid of the coefficient profile.  Its support, oldest represented
point, and storage requirement remain unrestricted.  The class may have
infinite support and may require a growing state.  For a bi-infinite input
sequence in \(\B\), let
\[
 x_{t-1}^{w}=\sum_{j\geq0}w_jp_{t-1-j},
 \qquad
 v_t^{w}=p_t-x_{t-1}^{w}.
\]
Thus \(v_t^w\) is the insertion vector from the weighted point
\(x_{t-1}^w\) to the new point \(p_t\).
The corresponding asymptotic adversarial value is
\begin{equation}
 \A_{d,\alpha}[w]
 =\limsup_{n\to\infty}\sup_{(p_t)_{t\in\mathbb Z}\subset\B}
   \frac1n\sum_{t=1}^n\norm{v_t^w}^{\alpha}.
 \label{eq:weighted-adversarial-value}
\end{equation}
The bi-infinite formulation makes the coefficients independent of the starting
time.  As in Eq.\eqref{eq:adversarial-value}, for each \(n\) the supremum is
taken before the limsup.
Associate with the weights the functions
\begin{equation}
 W(z)=\sum_{j\geq0}w_jz^j,
 \qquad
 F_w(z)=1-zW(z),
 \qquad
 m(w)=\max_{|z|=1}|F_w(z)|.
 \label{eq:weighted-generating-functions}
\end{equation}
The series for \(W\) converges uniformly on the closed unit disk.  Hence
\(F_w\) is continuous on the unit circle and the maximum defining \(m(w)\) is
attained.  The first-moment assumption also gives
\[
 F_w(e^{\mathrm{i}\omega})=-\mathrm{i}(a+1)\omega+o(|\omega|)
 \qquad (\omega\to0).
\]
Thus comparison at the same mean lag holds fixed both the temporal centroid
and the first-order response of the insertion vector to a slowly rotating
input.  Unit-circle maxima and Parseval's identity provide the corresponding
\(\ell_2\) convolution norm; see Chapter~I, Section~5 of
Katznelson~\cite{Katznelson2004} for the Parseval isometry used in
Theorem~\ref{thm:weighted-unit-circle}.

Optimization of unit-circle maxima for finite coefficient sequences also
appears in frequency-response design.  Convex formulations for
magnitude constraints are developed by Wu, Boyd, and Vandenberghe~\cite{WuBoydVandenberghe1998},
and Liu and Bauer~\cite{LiuBauer2010} study the frequency-domain limitations imposed by
nonnegative coefficients.  Here the coefficients are additionally normalized
as an averaging rule, their first moment is prescribed, and the objective is
the adversarial Euclidean insertion cost.  The analytic lower bound and
approximation guarantee characterize the exponential profile within this
geometric model.
For the exponential weights \(w_j=(1-\gamma)\gamma^j\),
\[
 F_\gamma(z)=F_w(z)=\frac{1-z}{1-\gamma z},\qquad
 a=\frac{\gamma}{1-\gamma},\qquad
 m(w)=\frac2{1+\gamma}.
\]

The following lemma relates the bi-infinite definition to a one-sided
construction.

\begin{lemma}
\label{lem:weighted-startup}
Let \(T_t=\sum_{j\geq t}w_j\) and, for \(t\geq1\), define
\[
 \widehat x_{t-1}^{w}
 =\sum_{j=0}^{t-1}w_jp_{t-1-j}+T_t p_0,
 \qquad
 \widehat v_t^w=p_t-\widehat x_{t-1}^{w}.
\]
Then, for every \(\alpha>0\),
\[
 \A_{d,\alpha}[w]
 =\limsup_{n\to\infty}\sup_{p_0,\ldots,p_n\in\B}
   \frac1n\sum_{t=1}^n\norm{\widehat v_t^w}^\alpha.
\]
For \(w_j=(1-\gamma)\gamma^j\), the point
\(\widehat x_{t-1}^{w}\) is exactly the state \(x_{t-1}\) generated by
Eq.\eqref{eq:recursion} from \(x_0=p_0\).
\end{lemma}

\begin{proof}
Extend a one-sided sequence by setting \(p_t=p_0\) for \(t<0\).  Its
bi-infinite weighted point is exactly \(\widehat x_{t-1}^{w}\), so its
asymptotic value is at most \(\A_{d,\alpha}[w]\).

Conversely, start from any bi-infinite sequence and replace its negative-index
points by \(p_0\).  At insertion \(t\), the two weighted points differ by at
most \(2T_t\).  Lemma~\ref{lem:power-difference} bounds the corresponding
difference of costs by
\[
 \begin{cases}
  (2T_t)^\alpha, & 0<\alpha\leq1,\\
  \alpha2^\alpha T_t, & \alpha\geq1.
 \end{cases}
\]
Since \(T_t\) tends to zero, the Cesaro mean of \(T_t^\alpha\) tends to zero.
For \(\alpha\geq1\), the stronger identity
\(\sum_{t\geq1}T_t=\sum_{j\geq0}jw_j=a\) applies.  Thus the difference of the
 two mean costs over the first \(n\) insertions tends to zero uniformly over all inputs.  Taking
suprema and then limsups proves the equality.  For exponential weights,
\(T_t=\gamma^t\); direct expansion of Eq.\eqref{eq:recursion} gives the final
statement.
\end{proof}

The next result evaluates the adversarial cost of every weight sequence through
the second power.

\begin{theorem}
\label{thm:weighted-unit-circle}
Suppose that \(d\geq2\) and that \(w\) satisfies
Eq.\eqref{eq:weighted-mean-lag}.  Then
\begin{equation}
 \A_{d,\alpha}[w]=m(w)^\alpha,
 \qquad 0<\alpha\leq2.
 \label{eq:weighted-unit-circle-exact}
\end{equation}
For every \(\alpha>0\),
\begin{equation}
 \A_{d,\alpha}[w]\geq m(w)^\alpha.
 \label{eq:weighted-unit-circle-lower}
\end{equation}
\end{theorem}

\begin{proof}
For \(L\geq0\), truncate the weighted point and the insertion vector by setting
\[
 x_{t-1}^{w,L}=\sum_{j=0}^Lw_jp_{t-1-j},
 \qquad
 v_t^{w,L}=p_t-x_{t-1}^{w,L},
\]
and let
\[
 F_{w,L}(z)=1-\sum_{j=0}^Lw_jz^{j+1},
 \qquad
 m_L(w)=\max_{|z|=1}|F_{w,L}(z)|.
\]
Fix the points needed for the insertions from \(1\) to \(n\), and set all
points outside the index interval \([-L,n]\) equal to zero.  Define the
vector-valued trigonometric polynomial
\[
 P(\omega)=\sum_{t=-L}^n p_te^{-\mathrm{i}t\omega}.
\]
The bilateral sequence
\(v_t^{w,L}=p_t-\sum_{j=0}^Lw_jp_{t-1-j}\) has Fourier series
\[
 F_{w,L}(e^{-\mathrm{i}\omega})P(\omega).
\]
Parseval's identity, applied coordinatewise on the complete bilateral
sequence, therefore gives
\begin{align*}
 \sum_{t\in\mathbb Z}\norm{v_t^{w,L}}^2
 &=\frac1{2\pi}\int_{-\pi}^{\pi}
   |F_{w,L}(e^{-\mathrm{i}\omega})|^2\norm{P(\omega)}^2\,d\omega\\
 &\leq m_L(w)^2\sum_{t=-L}^n\norm{p_t}^2
 \leq(n+L+1)m_L(w)^2.
\end{align*}
Restricting the sum on the left to \(1\leq t\leq n\) gives the required
finite upper bound.
Moreover,
\[
 \norm{v_t^w-v_t^{w,L}}\leq\sum_{j>L}w_j.
\]
Both insertion vectors have norm at most two, so the difference between their
squared norms is at most four times the right-hand side.  Since \(w\) is
summable, \(F_{w,L}\) converges uniformly to \(F_w\) on the unit circle and
the omitted mass tends to zero.  First let \(n\) tend to infinity and then let
\(L\) tend to infinity.  We obtain
\[
 \limsup_{n\to\infty}\sup_{(p_t)}
 \frac1n\sum_{t=1}^n\norm{v_t^w}^2
 \leq m(w)^2.
\]
For \(0<\alpha\leq2\), monotonicity of power means gives the corresponding
upper bound in Eq.\eqref{eq:weighted-unit-circle-exact}.

For the reverse inequality, identify the first two coordinates with the
complex plane.  For a fixed \(\omega\), use the rotating input
\[
 p_t=(\cos(\omega t),\sin(\omega t),0,\ldots,0).
\]
Every insertion vector then has norm \(|F_w(e^{-i\omega})|\).  Taking a value
of \(\omega\) that maximizes this expression proves
 Eq.\eqref{eq:weighted-unit-circle-lower} for every positive \(\alpha\), and proves
 equality when combined with the upper bound for \(0<\alpha\leq2\).
\end{proof}

The assumption \(d\geq2\) is used for the rotating input in the lower bound.
The Parseval upper bound is valid in every dimension.

Two optimization problems are involved.  For independent uniform input and a
fixed number \(N\geq2\) of insertions, Theorem~\ref{thm:finite-quadratic} gives the
unique parameter \(\gamma_N^*\) minimizing the expected total quadratic cost
among all \(\gamma\)-strategies.  The present subsection considers adversarial
input and the broader comparison class of all fixed weighted averages with a
prescribed mean lag.  Set
\[
 \operatorname{OPT}_{d,\alpha}(a)
 =\inf\{\A_{d,\alpha}[w]:w\text{ satisfies Eq.\eqref{eq:weighted-mean-lag}
 with mean lag }a\}.
\]
\medskip

\begin{theorem}
\label{thm:weighted-constant-approximation}
Let \(d\geq2\), \(a\geq0\), and
\[
 \gamma=\frac{a}{a+1},\qquad
 \kappa=\frac{2}{\pi(1-\log(\pi/2))}<1.161.
\]
For every \(0<\alpha\leq3\),
\begin{equation}
 \A_{d,\alpha}(\gamma)
 \leq \kappa^\alpha\operatorname{OPT}_{d,\alpha}(a).
 \label{eq:weighted-constant-approximation}
\end{equation}
 Moreover, for each fixed \(0<\alpha\leq3\), as \(a\) tends to infinity,
\begin{equation}
 \frac{\A_{d,\alpha}(\gamma)}
 {\operatorname{OPT}_{d,\alpha}(a)}
 \leq
 1+\frac{\alpha(1/2-1/\pi)}a+O(a^{-2}).
 \label{eq:weighted-asymptotic-approximation}
\end{equation}
\end{theorem}

For \(\alpha=2\), the exponential rule is therefore a
\(\kappa^2\)-approximation, with \(\kappa^2<1.348\), for every mean lag \(a\).
Because the exponential rule is one of the admissible weighted averages, its
ratio to \(\operatorname{OPT}_{d,2}(a)\) is at least one.  The second conclusion
of the theorem shows that this ratio converges to one as \(a\) tends to
infinity.  Hence the exponential rule is asymptotically optimal among all
nonnegative weight sequences with the same mean lag.
The constant \(\kappa\) follows from the uniform analytic bound in the proof.
Determining the best possible factor for finite \(a\) remains open.

\begin{proof}[Proof of Theorem~\ref{thm:weighted-constant-approximation}]
Lemma~\ref{lem:weighted-startup} identifies the \(\gamma\)-strategy with the
one-sided implementation of the fixed weights
\(w_j=(1-\gamma)\gamma^j\).  Their mean lag is
\(\gamma/(1-\gamma)=a\), which gives \(\gamma=a/(a+1)\).

We first derive a lower bound that holds for every weight sequence of mean lag
\(a\).  Let \(s=a+1\).  Since the incoming point \(p_t\) is one insertion step
after \(p_{t-1}\), \(s\) is the weighted mean number of insertion steps from a
 contributing point to \(p_t\).  Let \(M=m(w)\).  Since \(F_w(0)=1\), the
 submean inequality for the subharmonic function \(\log|F_w|\), equivalently
 Jensen's formula~\cite[Sec.~9.1.2]{Krantz1999}, gives, for \(0<r<1\),
\begin{equation}
 0\leq\frac1{2\pi}\int_{-\pi}^{\pi}
       \log|F_w(re^{i\omega})|\,d\omega.
 \label{eq:jensen-average}
\end{equation}
Zeros of \(F_w\) on the circle cause only integrable logarithmic singularities.
The displayed inequality follows by taking radii with no boundary zero and
passing to the limit, which is also the standard subharmonic interpretation of
the circular mean.
The identity
\[
 F_w(re^{i\omega})
 =\sum_{j\geq0}w_j
   \bigl(1-r^{j+1}e^{i(j+1)\omega}\bigr)
\]
and the inequalities
\[
 1-r^{j+1}\leq(j+1)(1-r),
 \qquad
 |1-e^{i(j+1)\omega}|\leq(j+1)|\omega|
\]
show that
\[
 |F_w(re^{i\omega})|
 \leq s(1-r+|\omega|).
\]
The maximum-modulus principle~\cite[Secs.~5.4.1--5.4.2]{Krantz1999} also gives
\(|F_w(re^{i\omega})|\leq M\).  Split the integral in
Eq.\eqref{eq:jensen-average} at an arbitrary \(b\in(0,\pi)\), use the first bound
on \(|\omega|\leq b\), and use the second bound on the remaining arc.  Letting
\(r\) tend to one is justified directly by
\[
 \int_0^b\log\bigl(s(1-r+\omega)\bigr)\,d\omega
 \longrightarrow b\bigl(\log(sb)-1\bigr).
\]
Indeed, the integral on the left equals
\[
 (b+1-r)\log\bigl(s(b+1-r)\bigr)
 -(1-r)\log\bigl(s(1-r)\bigr)-b.
\]
The resulting inequality is
\begin{equation}
 \log M
 \geq\frac{b[1-\log(sb)]}{\pi-b}.
 \label{eq:jensen-lower-bound}
\end{equation}

There is a unique \(y\in(0,1)\) satisfying
\[
 \pi sy e^{-y}=1.
\]
To verify existence and uniqueness, let \(h_s(y)=\pi sye^{-y}\).  This
function is strictly increasing on \((0,1)\), with \(h_s(0)=0\) and
\(h_s(1)=\pi s/e>1\).
Choose \(b=\pi y\) in Eq.\eqref{eq:jensen-lower-bound}.  Since
\(s\pi y=e^y\), this gives \(M\geq e^y\).  For exponential weighting with
mean lag \(a\), let
\[
 m_\gamma
 =1+\frac1{2a+1}=\frac{2s}{2s-1}.
\]
Using \(s=e^y/(\pi y)\), we obtain
\begin{equation}
 \frac{m_\gamma}{M}
 \leq\frac2{2e^y-\pi y}.
 \label{eq:gamma-weight-ratio}
\end{equation}
The denominator on the right is minimized at
\(y=\log(\pi/2)\), and its minimum is
\(\pi(1-\log(\pi/2))\).  Therefore
\[
 \frac{m_\gamma}{M}\leq \kappa.
\]

Theorem~\ref{thm:weighted-unit-circle} gives
\(\A_{d,\alpha}[w]\geq M^\alpha\) for every \(\alpha>0\), while
Theorem~\ref{thm:through-three} gives
\(\A_{d,\alpha}(\gamma)=m_\gamma^\alpha\) through
\(\alpha=3\).  Taking the infimum over \(w\) proves
Eq.\eqref{eq:weighted-constant-approximation}.

Finally, the defining equation for \(y\) gives, as \(s\) tends to infinity,
\[
 y=\frac1{\pi s}+O(s^{-2}),
 \qquad
 e^y=1+\frac1{\pi s}+O(s^{-2}),
 \qquad
 m_\gamma
 =1+\frac1{2s}+O(s^{-2}).
\]
Since \(s=a+1\), expanding \((m_\gamma/e^y)^\alpha\) gives
\[
 \left(\frac{m_\gamma}{e^y}\right)^\alpha
 =1+\frac{\alpha(1/2-1/\pi)}a+O(a^{-2}).
\]
The lower bound \(M\geq e^y\), followed by the two adversarial cost
identities used above, proves
Eq.\eqref{eq:weighted-asymptotic-approximation}.
\end{proof}

\section{Numerical evaluation}
\label{sec:computation}

This section evaluates the constants, the \(N\)-dependent calibrations, the
stationary radial laws, and the periodic adversarial constructions that enter
the theoretical results.  Online Resource~1 contains table-ready and
plot-ready CSV files for the numerical tables and figures, together with
recorded seeds and compact validation summaries.  Source code, build files,
complete calibration and adversarial grids, and regeneration instructions are
in the versioned repository cited in the Data and code availability statement.
Deterministic Gauss quadrature computes the radial and angular
integrals.  Monte Carlo experiments use recorded seeds, and a dense unit-circle
calculation evaluates the finite-support weight optimization, with a certified
continuous-maximum error bound.  The theorems in
the preceding sections establish the stated guarantees analytically.  The
calculations below quantify them and test their implementation over a broader
range of dimensions and parameters.

The C++ experiments and geometric validation were run using CGAL~6.1.1,
installed through vcpkg, in C++17 mode.  The versioned repository records the
build configuration, seed families, parameter grids, sample sizes, and
regeneration commands.

\subsection{Optimizer constants and calibration by input size}

Table~\ref{tab:optimizer-constants} evaluates Eq.\eqref{eq:H-constant}, the
optimizer in Eq.\eqref{eq:optimal-square-root-constant}, and the closed bounds
in Eq.\eqref{eq:closed-lambda-bounds}.  The values of \(\alpha\) range from
\(0.1\) to \(3\), with four representatives in the sublinear regime.  For
\(d=1\) and \(\alpha<1\), the hypothesis \(d+\alpha>2\) fails, so the
corresponding square-root calibrations are omitted.  The expanded dimension range shows that
\(\lambda^*_{d,1}\) decreases slowly toward its high-dimensional value, whereas
\(\lambda^*_{d,3}\) increases slowly.  The identity
\(\lambda^*_{d,2}=1\) is recovered to numerical precision for every displayed
dimension.  The complete grid in the repository contains
\(d=1,\ldots,100\), 195 values of
\(\alpha\) from \(0.01\) to \(20\), and 70 input sizes through \(10^7\).

For \(\alpha=2\), the numerical data in Online Resource~1 compare the unique exact
optimizer from Eq.\eqref{eq:optimal-gamma} with the first two terms of
Eq.\eqref{eq:optimal-gamma-asymptotic}.  The second-order approximation is
already accurate at \(N=16\).  These values provide a direct numerical check of
the scaled parameter and cost corrections in
Eqs.\eqref{eq:optimal-gamma-asymptotic} and
\eqref{eq:optimal-cost-asymptotic}.

The table also shows that the square-root calibration is numerically stable
across dimension.  Over the displayed range, dimension changes the optimal
coefficient much less than changing the power.  The exact quadratic identity
therefore separates a structural case from the two monotone trends seen for
\(\alpha=1\) and \(\alpha=3\).

Each displayed value of \(\lambda^*_{d,\alpha}\) gives the calibration
\(\gamma_N=1-\lambda^*_{d,\alpha}/\sqrt N\) directly.  Values below one make
the geometric weights decay more slowly than in the quadratic case, while
values above one make them decay faster.  Thus the table records both the
asymptotic coefficient and the parameter adjustment required at a specified
input size.

\begin{table}[!h]
\caption{Square-root calibration across dimensions.  The table gives the
coefficient \(H_{d,\alpha}\) in Eq.\eqref{eq:H-constant} and the minimizing
constant \(\lambda^*_{d,\alpha}\) in
Eq.\eqref{eq:optimal-square-root-constant}.  Dashes in the row \(d=1\) mark
sublinear parameter pairs outside the hypotheses of the square-root theorem.}
\label{tab:optimizer-constants}
\centering
\normalsize
\setlength{\tabcolsep}{3.0pt}
\begin{tabular*}{\textwidth}{@{\extracolsep{\fill}}rrrrrrrrrrr@{}}
\toprule
& \multicolumn{10}{c}{$H_{d,\alpha}$} \\
\cmidrule(l){2-11}
$d$ & $\alpha=0.1$ & $0.25$ & $0.5$ & $0.75$ & $1$ & $1.4$ & $1.8$ & $2$ & $2.6$ & $3$ \\
\midrule
1 & 0.0176 & 0.0434 & 0.0852 & 0.1261 & 0.1667 & 0.2318 & 0.2987 & 0.3333 & 0.4438 & 0.5250 \\
2 & 0.0249 & 0.0619 & 0.1229 & 0.1836 & 0.2445 & 0.3437 & 0.4466 & 0.5000 & 0.6709 & 0.7962 \\
3 & 0.0288 & 0.0720 & 0.1438 & 0.2161 & 0.2893 & 0.4093 & 0.5347 & 0.6000 & 0.8092 & 0.9627 \\
5 & 0.0329 & 0.0826 & 0.1663 & 0.2514 & 0.3385 & 0.4828 & 0.6348 & 0.7143 & 0.9698 & 1.1573 \\
10 & 0.0368 & 0.0927 & 0.1879 & 0.2861 & 0.3875 & 0.5574 & 0.7383 & 0.8333 & 1.1403 & 1.3659 \\
25 & 0.0395 & 0.0998 & 0.2036 & 0.3115 & 0.4239 & 0.6140 & 0.8181 & 0.9259 & 1.2755 & 1.5332 \\
50 & 0.0405 & 0.1025 & 0.2093 & 0.3209 & 0.4375 & 0.6354 & 0.8486 & 0.9615 & 1.3283 & 1.5989 \\
100 & 0.0410 & 0.1038 & 0.2123 & 0.3258 & 0.4446 & 0.6467 & 0.8648 & 0.9804 & 1.3564 & 1.6341 \\
\addlinespace[4pt]
& \multicolumn{10}{c}{$\lambda^*_{d,\alpha}$} \\
\cmidrule(l){2-11}
$d$ & $\alpha=0.1$ & $0.25$ & $0.5$ & $0.75$ & $1$ & $1.4$ & $1.8$ & $2$ & $2.6$ & $3$ \\
\midrule
1 & $--$ & $--$ & $--$ & $--$ & 1.0000 & 0.9967 & 0.9979 & 1.0000 & 1.0120 & 1.0247 \\
2 & 0.9977 & 0.9949 & 0.9914 & 0.9895 & 0.9890 & 0.9910 & 0.9962 & 1.0000 & 1.0159 & 1.0303 \\
3 & 0.9802 & 0.9795 & 0.9792 & 0.9800 & 0.9820 & 0.9872 & 0.9951 & 1.0000 & 1.0186 & 1.0342 \\
5 & 0.9602 & 0.9617 & 0.9649 & 0.9689 & 0.9736 & 0.9827 & 0.9937 & 1.0000 & 1.0220 & 1.0393 \\
10 & 0.9395 & 0.9432 & 0.9497 & 0.9568 & 0.9643 & 0.9775 & 0.9921 & 1.0000 & 1.0259 & 1.0453 \\
25 & 0.9236 & 0.9288 & 0.9378 & 0.9472 & 0.9569 & 0.9733 & 0.9908 & 1.0000 & 1.0294 & 1.0507 \\
50 & 0.9176 & 0.9233 & 0.9332 & 0.9434 & 0.9540 & 0.9716 & 0.9903 & 1.0000 & 1.0308 & 1.0529 \\
100 & 0.9144 & 0.9204 & 0.9307 & 0.9414 & 0.9524 & 0.9707 & 0.9900 & 1.0000 & 1.0316 & 1.0541 \\
\bottomrule
\end{tabular*}

\end{table}
\FloatBarrier

\FloatBarrier
\subsection{Uniform input}

Figure~\ref{fig:uniform-comparison} reports mean insertion cost for six
combinations of dimension and power, including \(\alpha<1\).  The exponential
rule uses \(\gamma_N=1-\lambda^*_{d,\alpha}/\sqrt N\), except that the exact
finite-\(N\) optimizer is used for \(\alpha=2\).  Before \(p_i\) is inserted,
the prefix-centroid rule uses
\(\bar p_{i-1}=i^{-1}\sum_{j=0}^{i-1}p_j\) as its attachment point.  Each point
in the figure averages between 30 and 6,000 independent trials, according to
\(N\); the shaded regions are normal-approximation 95 percent Monte Carlo
intervals.  Horizontal
lines give the exact expected costs of the center star and the input-order
path.  The separate vertical scales use the plotting area to expose the
finite-size differences in each parameter regime.

\begin{figure}[!h]
\centering
\includegraphics[width=\textwidth]{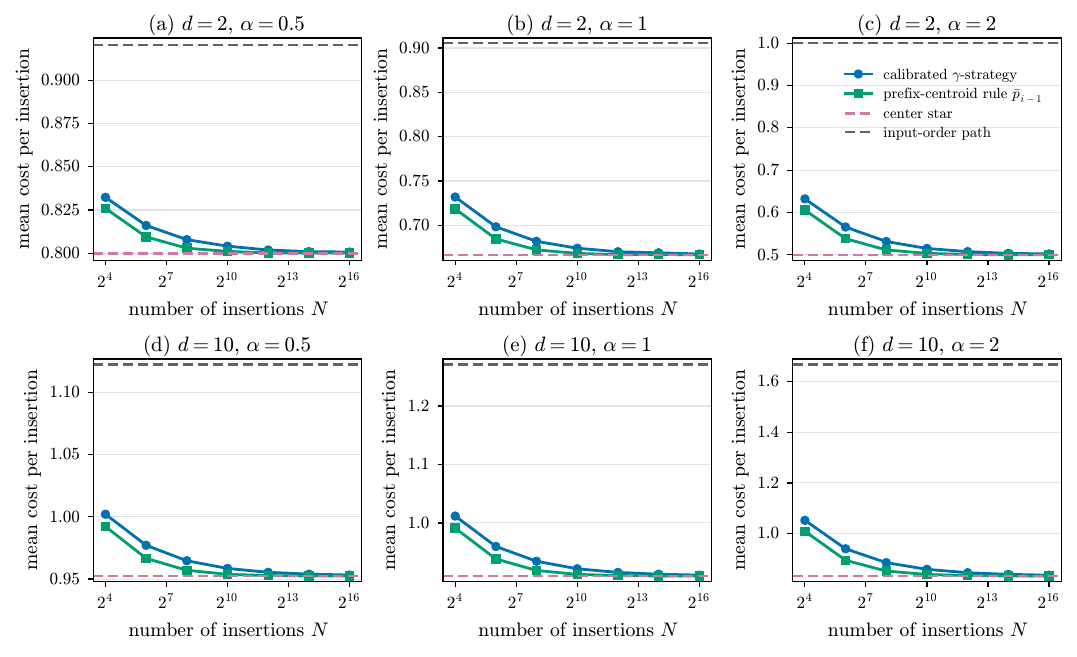}
\caption{Uniform-input comparison for \(d\in\{2,10\}\) and
\(\alpha\in\{0.5,1,2\}\).  The upper and lower rows correspond to \(d=2\)
and \(d=10\), respectively.  Points are Monte Carlo means, shaded regions are
normal-approximation 95 percent Monte Carlo intervals, and horizontal dashed
lines are exact expected costs.  The
prefix centroid before insertion \(i\) is
\(\bar p_{i-1}=i^{-1}\sum_{j=0}^{i-1}p_j\).}
\label{fig:uniform-comparison}
\end{figure}
\FloatBarrier

The comparison exhibits the intended interpolation.  At the displayed input
sizes, the calibrated exponential rule approaches the center-star level from
above and stays well below the input-order path, while retaining a
constant-memory update.  The finite-size decrease is most pronounced for
\(\alpha=2\), where the exact calibration is available, and it remains visible
for ordinary Euclidean length and for the sublinear power \(\alpha=0.5\).  The
prefix centroid is a useful numerical comparator because it uses the entire
preceding sequence; its different memory scale is reflected in the separation
between its curve and the exponential curve.

\FloatBarrier
\subsection{Stationary radius across dimensions}
\label{subsec:numerical-radius}

Radial symmetry permits an exact one-dimensional simulation of the stationary
state radius in every dimension.  Conditional on \(R=\norm{x_i}\), let
\(\rho=\norm{p_{i+1}}\) and let \(T\) be the scalar product between independent
uniform directions on the unit sphere.  The update satisfies
\begin{equation}
 R_{\mathrm{next}}^2
 =\gamma^2R^2+(1-\gamma)^2\rho^2
  +2\gamma(1-\gamma)R\rho T.
 \label{eq:radial-simulation}
\end{equation}
Here \(\Pr(\rho\leq r)=r^d\).  For \(d\geq2\),
\((T+1)/2\) has the beta distribution with both parameters equal to
\((d-1)/2\); for \(d=1\), \(T\) takes the values \(-1\) and \(1\) with equal
probability.  Equation~\eqref{eq:radial-simulation} therefore avoids storing
\(d\)-dimensional vectors while preserving the exact transition law.

For fixed \(\gamma\), the stationary distribution does not depend on
\(\alpha\).  Figure~\ref{fig:stationary-radius-dimension} therefore holds
\(\gamma\) fixed within each panel and varies it across nine panels.  It
contains 5 million retained observations for each of the 450 displayed pairs,
giving 2.25 billion observations over \(d=1,\ldots,50\).  Each horizontal row
is divided into 500 radial bins.  Its color is the empirical density divided
by the maximum density in that row, so the blue-to-red scale runs from zero to
one and reveals the location of the radial mass at every dimension.  The
superposed curve is the exact root mean square radius
\[
 \left(\frac{d}{d+2}\frac{1-\gamma}{1+\gamma}\right)^{1/2}
\]
from Eq.\eqref{eq:x-second}.  The simulation is checked against the exact
second moment for every displayed row.  The largest absolute discrepancy is
\(1.71\times10^{-4}\), and the largest relative discrepancy is \(0.214\%\).

\begin{figure}[!h]
\centering
\includegraphics[width=0.98\textwidth]{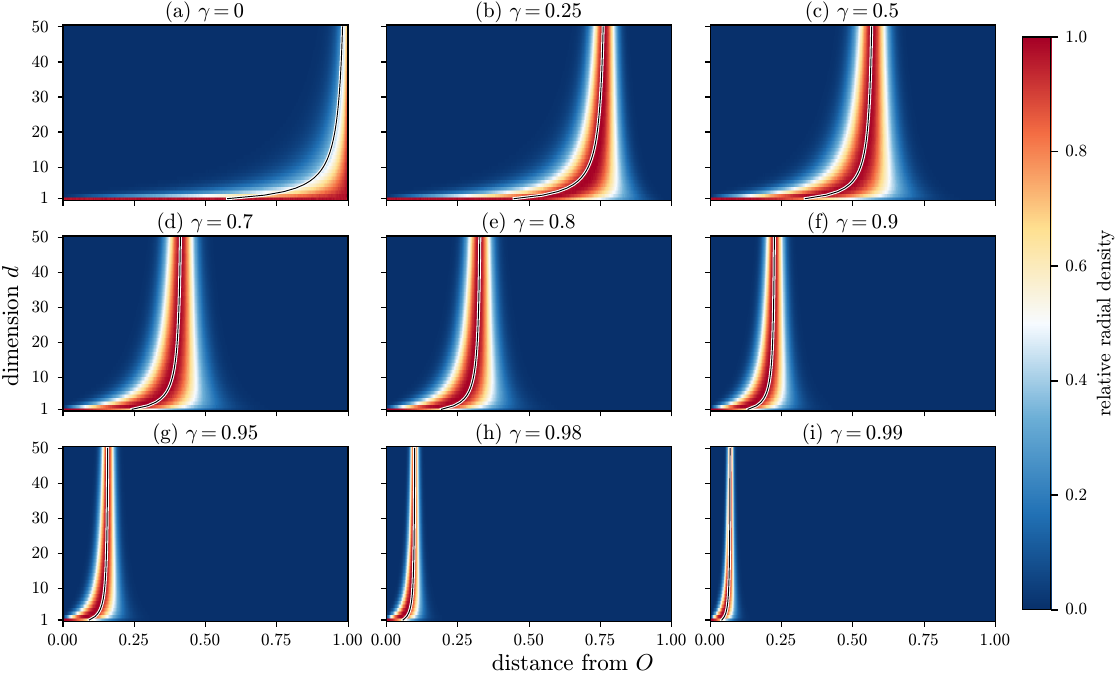}
\caption{Empirical stationary distribution of \(\norm{x_\gamma}\) for
\(\gamma\in\{0,0.25,0.5,0.7,0.8,0.9,0.95,0.98,0.99\}\).  The color in each
row is its radial density normalized to \([0,1]\), and the thin curve gives
the exact root mean square radius.  Each panel uses a fixed value of
\(\gamma\), and the dimension increases from 1 to 50 along its vertical
axis.}
\label{fig:stationary-radius-dimension}
\end{figure}
\FloatBarrier

For each fixed \(\gamma\), the high-density band becomes narrower as the
dimension increases, consistently with concentration of the uniform input
toward the boundary of the ball.  Increasing \(\gamma\) shifts the entire band
toward the center, and the root mean square curve follows that shift throughout
the displayed range.  These two effects are distinct: dimension sharpens the
radial profile, whereas the memory parameter controls its location.

\FloatBarrier
\subsection{Adversarial periodic sequences}
\label{subsec:numerical-adversarial}

Theorem~\ref{thm:through-three} proves, in every dimension, that antipodal
alternation attains the adversarial value for \(0<\alpha\leq3\).  We checked the
periodic-orbit implementation on 3.6 million seeded random sequences with periods
in \(\{2,3,4,6,8\}\), dimensions \(1,2,3,10,100\), four values of \(\gamma\),
and six values of \(\alpha\).  Through \(\alpha=3\), the largest observed ratio
to the alternating cost was \(1+4\times10^{-16}\).  This is a numerical
consistency check of the implementation over the sampled periodic families.

The block family in Proposition~\ref{prop:blocks} gives a more informative
diagnostic beyond the proved range.  Here a block of length \(m\) means \(m\)
successive copies of a unit vector followed by \(m\) successive copies of its
antipode, with this pattern repeated periodically.  Since the sequence
lies on a diameter, it is available in every dimension.  Figure~\ref{fig:block-transition}
maximizes its exact cost over \(1\leq m\leq250\), for 99 values of \(\gamma\)
and 4,000 values of \(\alpha\).  The ratio remains one throughout the range
\(0<\alpha\leq3\).  Above three, the length-two block is the first competitor
to exceed alternation for each value in Table~\ref{tab:block-transition};
longer blocks subsequently give much larger ratios.  For example, the crossing
occurs at \(\alpha\simeq3.2556\) when \(\gamma=0.5\), and at
\(\alpha\simeq3.0056\) when \(\gamma=0.9\).  This calculation illustrates the
change of worst-case behavior beyond the cubic range and leaves the complete
optimization problem for \(\alpha>3\) open.

\begin{figure}[!h]
\centering
\includegraphics[width=0.98\textwidth]{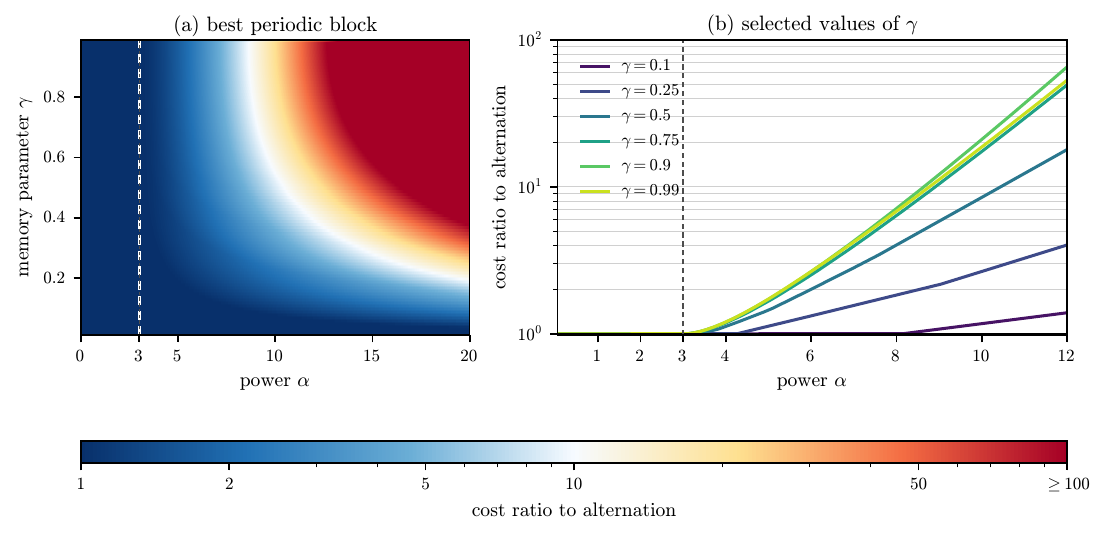}
\caption{Periodic block costs relative to antipodal alternation.  For the same
\((\alpha,\gamma)\), the plotted ratio is the mean cost of the best periodic
block with \(1\leq m\leq250\), divided by the mean cost of the alternating
sequence \(m=1\).  Panel (a) shows the full parameter grid; the dashed line
marks \(\alpha=3\).  Panel (b) shows sections for selected values of
\(\gamma\).  Every sequence lies on a diameter, so the calculation is
dimension-independent.}
\label{fig:block-transition}
\end{figure}
\FloatBarrier

The left panel makes the boundary at \(\alpha=3\) visible over the full
parameter range.  The sections in the right panel show that the departure from
alternation is continuous in cost and can be rapid in the maximizing block
length.  The transition also moves toward three as \(\gamma\) increases, a
pattern consistent with the critical values reported below.

\begin{table}[!h]
\caption{Transition within the periodic block family.  The second column is
the first power at which a block with \(m\leq250\) has larger mean cost than
alternation.  For \(a\in\{4,6,8,10\}\), \(m^*(a)\) is the maximizing block length
and \(R^*(a)\) is its mean cost divided by the alternating mean cost.}
\label{tab:block-transition}
\centering
\normalsize
\setlength{\tabcolsep}{3.8pt}
\begin{tabular*}{\textwidth}{@{\extracolsep{\fill}}rrrrrrrrrr@{}}
\toprule
$\gamma$ & transition $\alpha$ & $m^*(4)$ & $R^*(4)$ & $m^*(6)$ & $R^*(6)$ & $m^*(8)$ & $R^*(8)$ & $m^*(10)$ & $R^*(10)$ \\
\midrule
0.10 & 8.1203 & 1 & 1.0000 & 1 & 1.0000 & 1 & 1.0000 & 2 & 1.1740 \\
0.25 & 4.2480 & 1 & 1.0000 & 2 & 1.3261 & 2 & 1.8350 & 3 & 2.6586 \\
0.50 & 3.2556 & 3 & 1.1235 & 4 & 2.0107 & 5 & 4.0230 & 5 & 8.4865 \\
0.75 & 3.0418 & 7 & 1.1870 & 10 & 2.5153 & 11 & 6.3826 & 12 & 17.4155 \\
0.90 & 3.0056 & 18 & 1.2027 & 26 & 2.6524 & 30 & 7.1309 & 34 & 21.0437 \\
0.99 & 3.0001 & 187 & 1.2049 & 250 & 2.6594 & 250 & 6.8259 & 250 & 18.6879 \\
\bottomrule
\end{tabular*}

\end{table}
\FloatBarrier

Across the reported memory parameters, the first improving member of the block
family has length two.  At larger powers, the maximizing length grows and the
ratio can become substantially larger than one.  Thus the transition column
locates the loss of optimality of alternation within this family, while the
remaining columns quantify how quickly the competing periodic structure
strengthens.

\FloatBarrier
\subsection{Fixed-weight comparison}

The last calculation restricts Eq.\eqref{eq:weighted-mean-lag} to finite-support
weights.  We compute 24 values of the prescribed mean lag between 0.1 and 12
and display 12 representative values in Table~\ref{tab:finite-support-weights}.
The support endpoint is \(\max\{12,\lceil3a\rceil\}\).  For each prescribed mean
lag, \(|F_w|\) is minimized on a grid of 1,025 frequencies and evaluated on a
grid of 65,537 equally spaced frequencies in \([0,\pi]\).  Let
\(a_w=\sum_jjw_j\) for a computed profile.  If \(\widehat m\) is the largest
evaluated value and \(h=\pi/65{,}536\), then
\[
 \widehat m\leq m(w)\leq \widehat m+\frac{a_w+1}{2}h.
\]
Indeed, \(|dF_w(e^{\mathrm{i}\omega})/d\omega|\leq a_w+1\), and every frequency
lies within \(h/2\) of the grid.  The maximum discrepancy between \(a_w\) and
the prescribed value \(a\) is \(1.7\times10^{-13}\).  Table~\ref{tab:finite-support-weights}
reports this interval and the corresponding interval
\([b_-(a),b_+(a)]\) for \(b(a)=m_\gamma/m(w)\).  The last column is the
quantity to inspect: values above one exhibit a finite-support profile with
smaller adversarial cost.  For \(d\geq2\) and \(0<\alpha\leq2\), feasibility
of this profile and Theorem~\ref{thm:weighted-constant-approximation} bracket
the unrestricted approximation ratio between \(b_-(a)^\alpha\) and
\(\kappa^\alpha<1.161^\alpha\).

\begin{table}[!h]
\caption{Numerical finite-support comparison under the prescribed mean-lag
constraint.  The third column is a certified interval for
\(m(w)=\max_{|z|=1}|1-z\sum_jw_jz^j|\) for the computed profile.  The
exponential maximum uses weights with the same mean lag.  The last column is
the corresponding interval for \(b(a)=m_\gamma/m(w)\).  Values above one
measure a concrete improvement over exponential weights; for
\(0<\alpha\leq2\), \(b_-(a)^\alpha\) is a certified lower bound on the
approximation ratio, whose analytic upper bound is
\(\kappa^\alpha<1.161^\alpha\).}
\label{tab:finite-support-weights}
\centering
\normalsize
\setlength{\tabcolsep}{6.0pt}
\begin{tabular*}{\textwidth}{@{\extracolsep{\fill}}ccccc@{}}
\toprule
mean lag $a$ & support & certified $m(w)$ & exponential max. & factor \\
\midrule
0.1 & 0--12 & [1.79999, 1.80003] & 1.83333 & [1.01850, 1.01852] \\
0.25 & 0--12 & [1.56135, 1.56139] & 1.66667 & [1.06742, 1.06745] \\
0.5 & 0--12 & [1.39615, 1.39620] & 1.50000 & [1.07435, 1.07438] \\
0.75 & 0--12 & [1.30756, 1.30761] & 1.40000 & [1.07066, 1.07070] \\
1 & 0--12 & [1.25186, 1.25192] & 1.33333 & [1.06503, 1.06508] \\
1.5 & 0--12 & [1.18543, 1.18550] & 1.25000 & [1.05440, 1.05447] \\
2 & 0--12 & [1.14702, 1.14711] & 1.20000 & [1.04611, 1.04619] \\
3 & 0--12 & [1.10444, 1.10455] & 1.14286 & [1.03468, 1.03478] \\
4 & 0--12 & [1.09275, 1.09287] & 1.11111 & [1.01669, 1.01681] \\
6 & 0--18 & [1.06408, 1.06426] & 1.07692 & [1.01190, 1.01207] \\
8 & 0--24 & [1.04895, 1.04918] & 1.05882 & [1.00919, 1.00941] \\
12 & 0--36 & [1.03328, 1.03361] & 1.04000 & [1.00619, 1.00650] \\
\bottomrule
\end{tabular*}

\end{table}
\FloatBarrier

The largest base factor occurs at \(a=0.5\), with interval
\([1.07435,1.07438]\).  For quadratic cost, its lower endpoint gives a
\(15.4\%\) improvement over exponential weights.  The upper endpoint is below
\(1.01\) by \(a=8\) and decreases further at \(a=12\).  Thus the table locates
the largest short-memory gain and shows its decay toward the asymptotic
near-optimality in Eq.\eqref{eq:weighted-asymptotic-approximation}.  Online
Resource~1 contains all 24 rows and the computed weight vectors.

\FloatBarrier
\section{Concluding remarks}

The \(\gamma\)-strategy combines two properties of the endpoint
constructions.  Under uniform input, increasing \(\gamma\) improves every
positive moment of the stationary insertion length.  Under arbitrary input
sequences, for each fixed \(0\leq\gamma<1\), its exact asymptotic mean cost for
\(0<\alpha\leq3\) approaches the
 unit worst-case scale of the center-\(O\) star as \(\gamma\) increases to one.  The
attachment point remains a geometrically weighted average of the observed
points.  These results formalize the two properties that motivate the rule:
 its expected cost under uniform sampling retains the influence of the observed
 sequence, while its adversarial cost remains explicitly controlled.

For ordinary Euclidean length, the optimized rule has the same leading
 expected cost as the center-\(O\) star and a strictly smaller leading constant than
the input-order path and the star centered at \(p_0\).  For \(N\) insertions,
Theorems~\ref{thm:square-root-constant} and
\ref{thm:global-finite-optimizer} determine the asymptotic location of every
minimizing constant parameter and its \(\sqrt N\) correction for general
\(\alpha\) under the stated regularity condition.  The quadratic case has a
unique minimizer and an exact expression for every \(N\).
Corollary~\ref{cor:same-N} shows that the same square-root calibration has
the sharp distributional correction and adversarial mean cost
\(1+O(N^{-1/2})\) for the same value of \(N\).  Propositions~\ref{prop:running-average} and
\ref{prop:running-average-adversarial} clarify the role of the
constant-parameter restriction: averaging all preceding points gives a
\(\log N\) uniform-input correction and adversarial value one through the
quadratic range, with a mean lag that grows with the sequence.  The fixed
exponential rule supplies a time-homogeneous coefficient profile and the exact
cubic-range guarantee.

The comparison with general fixed nonnegative weighted averages provides a
second justification for the exponential weights.  Among all fixed
nonnegative weight sequences with the same mean lag, their adversarial cost is
within the factor \(\kappa^\alpha\) of
Theorem~\ref{thm:weighted-constant-approximation}, and the
sharper ratio in Eq.\eqref{eq:weighted-asymptotic-approximation} converges to
one as the mean lag grows.  This comparison class may use unbounded support and
 storage, as explained in Section~\ref{sec:weighted-adversarial}.  The
 reproducible computations in Section~\ref{sec:computation} quantify the
 calibration, the stationary radial concentration across dimensions, the
 transition within the periodic block family, and the finite-support
 comparison.  The appendices
 collect the stationary transform, higher moments, and endpoint limits that
 support the distributional analysis without obscuring the three main results.

 For every fixed \(0<\gamma<1\), periodic endpoint blocks produce a strictly
 larger adversarial mean cost than alternation at sufficiently high powers.
 Determining the complete adversarial value for \(\alpha>3\), including the
 smallest power at which alternation ceases to be optimal, remains open.
 Further questions include other convex bodies,
 nonuniform input distributions, and time-varying weighting rules with
 controlled mean lag.

\section*{Data and code availability}

Table-ready and plot-ready data for the numerical tables and figures, recorded
seeds, and compact validation summaries are provided in the supplementary
archive distributed with the versioned release of the project repository.
Source code, build files, and complete calibration and adversarial grids are
archived in the
\href{https://github.com/tashimir/sequential-euclidean-connections-exponential-memory}{versioned project repository}.
The study uses no external datasets.

\appendix

\section{Stochastic comparison for the running mean}
\label{app:running-mean}

The cost estimate in Proposition~\ref{prop:running-average} is sufficient for
the comparison in the main text.  The next result records a stronger property
of the same attachment point.

\begin{proposition}
\label{prop:sample-mean-peakedness}
Fix \(t\geq0\), let \(\lambda_0,\ldots,\lambda_t\geq0\) satisfy
\(\sum_{j=0}^t\lambda_j=1\), and let
\[
 y_\lambda=\sum_{j=0}^t\lambda_jp_j.
\]
If \(p\) is an independent uniform point in \(\B\), then
\begin{equation}
 \Pr(\norm{p-\bar x_t}>r)
 \leq \Pr(\norm{p-y_\lambda}>r)
 \qquad(r\geq0).
 \label{eq:sample-mean-stochastic}
\end{equation}
Consequently, the expected \(\alpha\)-power of the insertion length is no
larger for \(\bar x_t\) than for \(y_\lambda\), for every \(\alpha>0\).
\end{proposition}

\begin{proof}
Every probability vector \((\lambda_0,\ldots,\lambda_t)\) majorizes the equal
vector \((1/(t+1),\ldots,1/(t+1))\).  The uniform density on \(\B\) is symmetric
and log-concave in the extended-value sense.  Lemma~\ref{lem:olkin-tong} gives
\[
 \Pr(\bar x_t\in K)\geq\Pr(y_\lambda\in K)
\]
for every compact convex set \(K\) symmetric about the origin.  For fixed
\(r\), let
\[
 h_r(z)=
 \frac{\operatorname{vol}(\B\cap B(z,r))}{\operatorname{vol}(\B)}
\]
for \(z\in\B\).  Lemma~\ref{lem:translated-ball-overlap} shows that \(h_r\) is
radial and nonincreasing, so its superlevel sets are balls centered at the
 origin.  Since \(0\leq h_r\leq1\), for any random vector \(y\) in \(\B\),
 \[
  \E h_r(y)=\int_0^1\Pr\bigl(h_r(y)>s\bigr)\,ds.
 \]
 Each set \(\{z:h_r(z)>s\}\) is either empty or a ball centered at the origin.
 Such an open ball is an increasing union of compact balls, so the peakedness
 inequality extends to it by continuity from below.  Applying that inequality
 inside the integral gives
 \(\E h_r(\bar x_t)\geq\E h_r(y_\lambda)\).  Averaging over the independent
point \(p\) gives Eq.\eqref{eq:sample-mean-stochastic}.  The moment comparison
follows from
\[
 \E X^\alpha
 =\int_0^\infty\alpha r^{\alpha-1}\Pr(X>r)\,dr.
\]
\end{proof}

\section{Exact transforms and fourth moments}
\label{app:stationary-transforms}

This appendix gives an exact transform description of the stationary state and
records the fourth moments used to check the high-memory limit in
Appendix~\ref{app:stationary-limits}.

\begin{lemma}
\label{lem:ball-transform}
If \(p\) is uniform in \(\B\), its characteristic function is
\[
 \psi_d(\norm\xi)
 =\E e^{\,\mathrm{i}\langle\xi,p\rangle}
 =2^{d/2}\Gamma\left(\frac d2+1\right)
   \frac{J_{d/2}(\norm\xi)}{\norm\xi^{d/2}},
 \qquad \psi_d(0)=1,
\]
where \(J_\nu\) denotes the Bessel function of the first kind and order
\(\nu\).
\end{lemma}

\begin{proof}
Rotation invariance makes the characteristic function depend only on
\(r=\norm\xi\).  Slicing the ball orthogonally to the first coordinate gives
\[
 \psi_d(r)
 =\frac{\Gamma(d/2+1)}
        {\sqrt{\pi}\,\Gamma((d+1)/2)}
   \int_{-1}^1 e^{\mathrm{i}rt}
        (1-t^2)^{(d-1)/2}\,dt .
\]
Poisson's integral representation for the Bessel function
\cite[Eq.~10.9.4]{DLMF} gives
\[
 J_\nu(r)
 =\frac{(r/2)^\nu}
        {\sqrt{\pi}\,\Gamma(\nu+1/2)}
   \int_{-1}^1 e^{\mathrm{i}rt}
        (1-t^2)^{\nu-1/2}\,dt
 .
\]
With \(\nu=d/2\), this yields the stated expression for \(\psi_d(r)\).  The
value at the origin follows from
\(J_\nu(r)\sim(r/2)^\nu/\Gamma(\nu+1)\).
\end{proof}

\begin{theorem}
\label{thm:product}
The characteristic function of \(x_\gamma\) is
\begin{equation}
 \E e^{\,\mathrm{i}\langle\xi,x_\gamma\rangle}
 =\prod_{j=0}^\infty
 \psi_d\bigl((1-\gamma)\gamma^j\norm\xi\bigr),
 \label{eq:bessel-product}
\end{equation}
and the product converges locally uniformly in \(\xi\).
If \(p\) is an additional independent uniform point, then
\begin{equation}
 \E e^{\,\mathrm{i}\langle\xi,p-x_\gamma\rangle}
 =\psi_d(\norm\xi)
  \prod_{j=0}^\infty
  \psi_d\bigl((1-\gamma)\gamma^j\norm\xi\bigr).
 \label{eq:edge-bessel-product}
\end{equation}
\end{theorem}

\begin{proof}
For
\[
 x_\gamma^{(n)}
 =(1-\gamma)\sum_{j=0}^n\gamma^jp_j,
\]
independence gives
\[
 \E e^{\,\mathrm{i}\langle\xi,x_\gamma^{(n)}\rangle}
 =\prod_{j=0}^n
 \psi_d\bigl((1-\gamma)\gamma^j\norm\xi\bigr).
\]
At the origin,
\[
 \psi_d(r)=1-\frac{r^2}{2(d+2)}+O(r^4).
\]
On every compact set of values of \(\xi\), the sum of
\(\lvert1-\psi_d((1-\gamma)\gamma^j\norm\xi)\rvert\) is bounded by a convergent
geometric series, after increasing the bound to cover finitely many initial
terms.  The finite products therefore converge locally uniformly.  Moreover,
\(\norm{x_\gamma-x_\gamma^{(n)}}\leq\gamma^{n+1}\), so the partial sums converge
almost surely and in every \(L^q\), \(q>0\).  Their characteristic functions
converge to that of \(x_\gamma\), which proves
Eq.\eqref{eq:bessel-product}.  The additional point is independent, and
\(\psi_d\) is even.  Multiplying the characteristic functions proves
Eq.\eqref{eq:edge-bessel-product}.
\end{proof}

For \(r>0\), let
\[
 s_r(\gamma)=\frac{(1-\gamma)^r}{1-\gamma^r}.
\]

\begin{proposition}
\label{prop:fourth-moments}
The fourth moments of the stationary state and insertion vector are
\begin{align}
 \E\norm{x_\gamma}^4
 &=\frac{d}{d+2}s_2(\gamma)^2
   -\frac{2d}{(d+2)(d+4)}s_4(\gamma),
 \label{eq:x-fourth}\\
 \E\norm{p-x_\gamma}^4
 &=\frac{d}{d+2}\bigl(1+s_2(\gamma)\bigr)^2
   -\frac{2d}{(d+2)(d+4)}\bigl(1+s_4(\gamma)\bigr),
 \label{eq:e-fourth}
\end{align}
where \(p\) is an additional independent uniform point.
\end{proposition}

\begin{proof}
For a uniform point \(p\) in \(\B\),
\[
 \E p=0,\qquad
 \E\norm p^2=\frac{d}{d+2},\qquad
 \E\norm p^4=\frac{d}{d+4},
\]
and isotropy gives \(\E pp^{\mathsf T}=I/(d+2)\).
First take a finite sum \(z_n=\sum_{j=0}^n b_jp_j\).  Expanding its fourth
power, only terms in which every index occurs at least twice survive.  For
\(j\neq k\),
\[
 \E\langle p_j,p_k\rangle^2
 =\operatorname{tr}\left(\frac{I}{d+2}\frac{I}{d+2}\right)
 =\frac{d}{(d+2)^2}.
\]
The contribution of a pair \(j<k\) is
\[
 b_j^2b_k^2
 \left\{
 2\left(\frac{d}{d+2}\right)^2
 +\frac{4d}{(d+2)^2}
 \right\}
 =\frac{2d}{d+2}b_j^2b_k^2.
\]
Adding the single-index terms gives
\begin{equation}
 \E\norm{z_n}^4
 =\frac{d}{d+2}\left(\sum_{j=0}^n b_j^2\right)^2
  -\frac{2d}{(d+2)(d+4)}\sum_{j=0}^n b_j^4.
 \label{eq:general-fourth}
\end{equation}
For \(x_\gamma^{(n)}\), take
\(b_j=(1-\gamma)\gamma^j\).  The deterministic tail bound
\(\norm{x_\gamma-x_\gamma^{(n)}}\leq\gamma^{n+1}\) gives convergence in
\(L^4\).  Letting \(n\) tend to infinity in
Eq.\eqref{eq:general-fourth} proves Eq.\eqref{eq:x-fourth}.
For \(p-x_\gamma\), first add the independent coefficient \(1\) to the finite
sum and change the other signs.  The same \(L^4\) limit proves
Eq.\eqref{eq:e-fourth}.
\end{proof}

\section{Endpoint limits of the stationary state and insertion length}
\label{app:stationary-limits}

For \(0\leq\gamma<1\), the stationary state has the series representation
\[
 x_\gamma=(1-\gamma)\sum_{j=0}^{\infty}\gamma^j p_j
\]
from Eq.\eqref{eq:stationary-series}.  This representation makes both endpoint
regimes explicit.  At \(\gamma=0\), \(x_0=p_0\) is uniform in \(\B\).  Under
the coupling provided by the same sequence \((p_j)_{j\geq0}\),
\begin{equation}
 \E\norm{x_\gamma-p_0}^2
 =\frac{2d}{d+2}\frac{\gamma^2}{1+\gamma}.
 \label{eq:low-memory-coupling}
\end{equation}
Indeed, the coefficients of \(x_\gamma-p_0\) are \(-\gamma\) for \(p_0\) and
\((1-\gamma)\gamma^j\) for \(p_j\), \(j\geq1\).  Independence and centering
therefore give
\[
 \E\norm{x_\gamma-p_0}^2
 =\frac{d}{d+2}\left(\gamma^2
 +(1-\gamma)^2\sum_{j\geq1}\gamma^{2j}\right),
\]
which is Eq.\eqref{eq:low-memory-coupling}.  In particular,
\(x_\gamma\to p_0\) in \(L^2\) as \(\gamma\to0^+\), and
\[
 x_\gamma=p_0+\gamma(p_1-p_0)+O_{L^2}(\gamma^2).
\]
For intermediate values of \(\gamma\), Theorem~\ref{thm:product} gives the
exact characteristic function, Proposition~\ref{prop:stationary-moment-bounds}
gives its second moment, and the peakedness comparison in the proof of
Theorem~\ref{thm:stochastic-monotonicity} shows that the state becomes more
concentrated about \(O\) as \(\gamma\) increases.  As \(\gamma\to1^-\), the
unscaled state converges to \(O\), while the following theorem identifies the
nondegenerate limit after rescaling.  The endpoint \(\gamma=1\) is distinct:
the recursion preserves its initial state and hence admits every initial law as
a stationary law.  The stationary distribution selected by
Eq.\eqref{eq:stationary-series} is consequently defined only for
\(\gamma<1\).

\begin{theorem}
\label{thm:gaussian-state}
As \(\gamma\) tends to one,
\begin{equation}
 \sqrt{\frac{1+\gamma}{1-\gamma}}\,x_\gamma
 \ \Longrightarrow\
 \mathcal N\left(0,\frac1{d+2}I_d\right).
 \label{eq:gaussian-state}
\end{equation}
Moreover,
\begin{equation}
 \frac{\E\norm{x_\gamma}^4}
      {(\E\norm{x_\gamma}^2)^2}
 =\frac{d+2}{d}
 -\frac{2(d+2)}{d(d+4)}\frac{1-\gamma^2}{1+\gamma^2}.
 \label{eq:radial-fourth-ratio}
\end{equation}
Thus the standardized radial fourth moment converges to
\((d+2)/d\), the value for an isotropic Gaussian vector.
\end{theorem}

\begin{proof}
After multiplication by the factor in Eq.\eqref{eq:gaussian-state}, the
coefficient of \(p_j\) is
\[
 b_j(\gamma)=\sqrt{1-\gamma^2}\,\gamma^j.
\]
These coefficients satisfy
\[
 \sum_{j\geq0}b_j(\gamma)^2=1,
 \qquad
 \sum_{j\geq0}b_j(\gamma)^4
 =\frac{1-\gamma^2}{1+\gamma^2}.
\]
Let \(y_\gamma\) denote the vector on the left of
Eq.\eqref{eq:gaussian-state}.  For fixed \(\xi\in\R^d\), independence and
Lemma~\ref{lem:ball-transform} give
\[
 \E e^{\,\mathrm{i}\langle\xi,y_\gamma\rangle}
 =\prod_{j=0}^\infty\psi_d\bigl(b_j(\gamma)\norm\xi\bigr).
\]
The expansion at the origin is
\[
 \log\psi_d(r)=-\frac{r^2}{2(d+2)}+O(r^4).
\]
Because \(\max_j b_j(\gamma)=\sqrt{1-\gamma^2}\) tends to zero, this expansion
applies uniformly to all factors for each fixed \(\xi\) when \(\gamma\) is
sufficiently close to one.  Consequently,
\begin{align*}
 \log\E e^{\,\mathrm{i}\langle\xi,y_\gamma\rangle}
 &=-\frac{\norm\xi^2}{2(d+2)}
   +O\left(\norm\xi^4\sum_{j\geq0}b_j(\gamma)^4\right)\\
 &=-\frac{\norm\xi^2}{2(d+2)}+o(1).
\end{align*}
The limiting function is the characteristic function of
\(\mathcal N(0,I_d/(d+2))\).  The continuity theorem for characteristic
functions in \(\R^d\) \cite[Sec.~29]{Billingsley1995} proves
Eq.\eqref{eq:gaussian-state}.  Dividing
Eq.\eqref{eq:x-fourth} by the square of Eq.\eqref{eq:x-second}, and using
\[
 \frac{s_4(\gamma)}{s_2(\gamma)^2}
 =\frac{1-\gamma^2}{1+\gamma^2},
\]
gives Eq.\eqref{eq:radial-fourth-ratio}.
\end{proof}

The upper block of Figure~\ref{fig:stationary-endpoint-interpolation}
illustrates both endpoint regimes in dimension two.  It displays one coordinate of
\(Y_\gamma=\sqrt{(1+\gamma)/(1-\gamma)}\,x_\gamma\).  As
\(\gamma\to0^+\), this variable converges to one coordinate of a uniform point
in the unit disk, with density
\begin{equation}
 f_0(y)=\frac{2}{\pi}\sqrt{1-y^2},\qquad |y|\leq1.
 \label{eq:disk-coordinate-density}
\end{equation}
As \(\gamma\to1^-\), Theorem~\ref{thm:gaussian-state} gives the density of
\(\mathcal N(0,1/4)\).  These two fixed curves make visible which endpoint
better describes each intermediate distribution.  The C++/CGAL implementation
retains 150 million observations after burn-in for each value of \(\gamma\).
The display histograms use 2,400 bins and light smoothing.  The diagnostics
\(D_0\) and \(D_1\) are binned CDF distances to the low-memory and high-memory
limits, respectively, evaluated on a separate partition with 1,048,576 bins.
They are trajectory diagnostics rather than iid test statistics or confidence
bounds.  The lower block uses the same numerical scale to display the distinct
endpoint interpolation of the insertion length.

We now turn from the state \(x_\gamma\) to the insertion length
\(Z_\gamma=\norm{p-x_\gamma}\), where \(p\) is an additional independent
uniform point in \(\B\).  Its two endpoints differ from those of the state.
Let \(\rho=\norm p\), where \(p\) is a uniform point in \(\B\).  Its density is
\(d r^{d-1}\) on \([0,1]\), and its distribution function is
\(\Pr(\rho\leq r)=r^d\).  For real-valued random variables \(X\) and \(Y\),
\(W_1(\mathcal L(X),\mathcal L(Y))\) denotes the infimum of
\(\E|X'-Y'|\) over all couplings \((X',Y')\) with the prescribed marginal
laws.

\begin{theorem}
\label{thm:radial-approximation}
The Wasserstein distance between the laws of \(Z_\gamma\) and \(\rho\)
satisfies
\begin{equation}
 W_1\bigl(\mathcal L(Z_\gamma),\mathcal L(\rho)\bigr)
 \leq
 \sqrt{\frac{d}{d+2}\frac{1-\gamma}{1+\gamma}}.
 \label{eq:wasserstein-radial}
\end{equation}
\end{theorem}

\begin{proof}
Use the coupling in which \(Z_\gamma\) and \(\rho\) share the point \(p\).
The reverse triangle inequality gives
\[
 |Z_\gamma-\rho|\leq\norm{x_\gamma}.
\]
The definition of \(W_1\), followed by Cauchy--Schwarz and
Eq.\eqref{eq:x-second}, gives
\[
 W_1\bigl(\mathcal L(Z_\gamma),\mathcal L(\rho)\bigr)
 \leq \E\norm{x_\gamma}
 \leq \sqrt{\E\norm{x_\gamma}^2}
 =\sqrt{\frac{d}{d+2}\frac{1-\gamma}{1+\gamma}}.
\]
\end{proof}

The opposite endpoint also has an explicit distribution.  For \(a,b>0\), let

\[
 I_t(a,b)=\frac{\int_0^t s^{a-1}(1-s)^{b-1}\,ds}
                 {\int_0^1 s^{a-1}(1-s)^{b-1}\,ds}
\]

denote the regularized incomplete beta function.

\begin{proposition}
\label{prop:zero-memory-law}
At \(\gamma=0\), \(Z_0\) is the distance between two independent uniform
points in \(\B\).  Its density is
\begin{equation}
 f_{Z_0}(r)
 =d r^{d-1} I_{1-r^2/4}\left(\frac{d+1}{2},\frac12\right),
 \qquad 0<r<2,
 \label{eq:zero-memory-density}
\end{equation}
and its distribution function is
\begin{equation}
 \Pr(Z_0\leq r)
 =r^d I_{1-r^2/4}\left(\frac{d+1}{2},\frac12\right)
  +I_{r^2/4}\left(\frac{d+1}{2},\frac{d+1}{2}\right),
 \qquad 0\leq r\leq2.
 \label{eq:zero-memory-cdf}
\end{equation}
In particular, when \(d=1\), the density is \(1-r/2\) and the distribution
function is \(r-r^2/4\) on \([0,2]\).
\end{proposition}

\begin{proof}
Let \(v_d\) be the volume of the unit ball.  If \(u\in\R^d\) and
\(r=\norm u\leq2\), the density of the difference of two independent uniform
points, evaluated at \(u\), is
\[
 \frac{\operatorname{vol}(\B\cap(\B+u))}{v_d^2}.
\]
Slicing the intersection orthogonally to the line through the two centers
gives
\[
 \frac{\operatorname{vol}(\B\cap(\B+u))}{v_d}
 =I_{1-r^2/4}\left(\frac{d+1}{2},\frac12\right).
\]
Multiplication by the radial Jacobian \(d v_d r^{d-1}\) proves
Eq.\eqref{eq:zero-memory-density}; this is the ball-distance distribution of
Hammersley~\cite{Hammersley1950}.  Integrating the density gives
Eq.\eqref{eq:zero-memory-cdf}.  Equivalently, differentiation of the
right-hand side of Eq.\eqref{eq:zero-memory-cdf}, together with the beta
duplication identity
\[
 2^d B\left(\frac{d+1}{2},\frac{d+1}{2}\right)
 =B\left(\frac12,\frac{d+1}{2}\right),
\]
recovers Eq.\eqref{eq:zero-memory-density}, and the value at \(r=0\) fixes the
integration constant.
\end{proof}

For \(d=2\), the two limiting densities used in the lower block of
Figure~\ref{fig:stationary-endpoint-interpolation} have simple forms.  At the
low-memory endpoint,
\begin{equation}
 f_{Z_0}(r)=\frac{4r}{\pi}\left(
 \arccos\frac r2-\frac r2\sqrt{1-\frac{r^2}{4}}
 \right),\qquad 0\leq r\leq2,
 \label{eq:disk-distance-density}
\end{equation}
whereas the high-memory limit \(\rho\) has density \(2r\) on \([0,1]\).
The same simulations used in the upper block
record \(Z_\gamma\) before every state update.  The figure uses the same eight
values of \(\gamma\), sample size, display resolution, and independent CDF
partition.  Thus its \(D_0\) and \(D_1\) values measure the approach to the
disk-distance and uniform-radius laws on a common numerical scale.

\begin{figure}[!htbp]
\centering
\includegraphics[width=\textwidth]{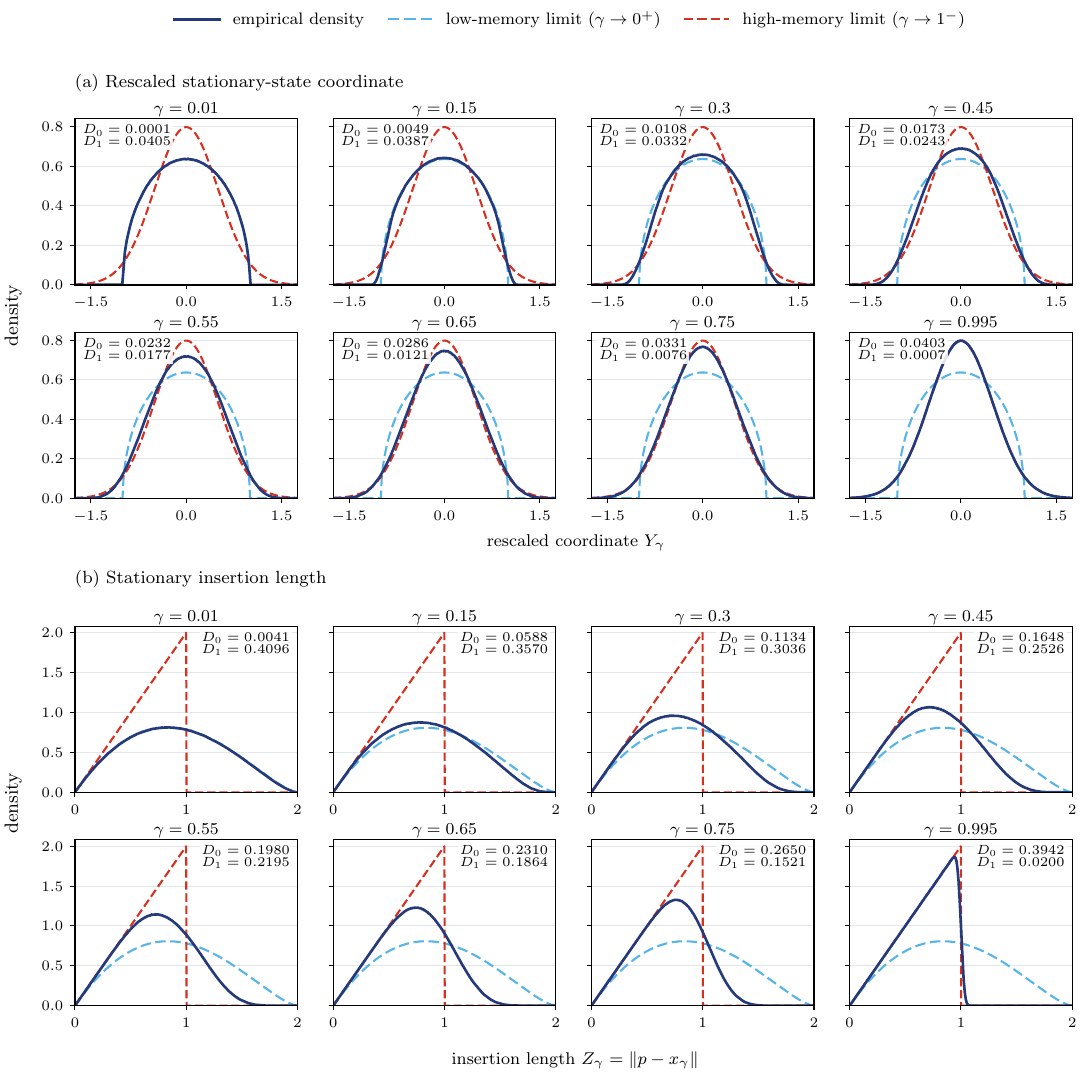}
\caption{Endpoint interpolation in dimension two, based on 150 million
retained observations for each displayed value of \(\gamma\).  The upper
block shows one-coordinate densities of
\(Y_\gamma=\sqrt{(1+\gamma)/(1-\gamma)}\,x_\gamma\), compared with the
uniform-disk coordinate law and the Gaussian limit.  The lower block shows
densities of \(Z_\gamma=\norm{p-x_\gamma}\), compared with the distance between
two independent uniform points in the disk and the uniform-radius limit.
Dark-blue solid curves are empirical densities, light-blue dashed curves are the
low-memory limits, and red dashed curves are the high-memory limits.  Each
panel reports the binned CDF distances \(D_0\) and \(D_1\) to these reference
laws.  The common ordering of \(\gamma\) makes the two interpolations directly
comparable.}
\label{fig:stationary-endpoint-interpolation}
\end{figure}

Proposition~\ref{prop:zero-memory-law},
Theorem~\ref{thm:stochastic-monotonicity}, and
Theorem~\ref{thm:radial-approximation} now describe the interpolation without
identifying the endpoint distributions with those of \(x_\gamma\).  At
\(\gamma=0\), \(Z_0\) has the ball-distance law in
Eq.\eqref{eq:zero-memory-cdf}.  As \(\gamma\) increases, its tail probabilities
decrease.  For intermediate values, the exact characteristic function of the
insertion vector is given by Eq.\eqref{eq:edge-bessel-product}.  Finally,
\[
 Z_\gamma\ \Longrightarrow\ \rho,
 \qquad
 \Pr(\rho\leq r)=r^d\quad(0\leq r\leq1),
 \qquad \gamma\to1^-.
\]
This last limit is one-sided.  At \(\gamma=1\), the state remains equal to its
initial value, so the corresponding insertion-length distribution depends on
that initialization.  In particular, initialization at \(O\) gives the radial
law immediately, while a uniformly distributed initial state gives the
ball-distance law at every step.

\section{Periodic endpoint blocks beyond the cubic range}
\label{app:periodic-blocks}

Theorem~\ref{thm:through-three} covers \(0<\alpha\leq3\).  The following family
shows that Eq.\eqref{eq:main-value} cannot remain valid for every
\(\alpha>3\).

\begin{proposition}
\label{prop:blocks}
Let \(0<\gamma<1\), and fix a unit vector \(u\).  Repeat \(m\) copies of
\(u\), followed by \(m\) copies of \(-u\), periodically.  The mean insertion
cost on the attracting periodic orbit is
\begin{equation}
 C_{m,\alpha}(\gamma)
 =\frac{2^\alpha}{m}
  \frac{1-\gamma^{\alpha m}}
  {(1-\gamma^\alpha)(1+\gamma^m)^\alpha}.
 \label{eq:block-cost}
\end{equation}
The case \(m=1\) is antipodal alternation.  For every fixed \(m\geq2\), the
inequality
\(C_{m,\alpha}(\gamma)>C_{1,\alpha}(\gamma)\) holds for all sufficiently large
\(\alpha\).  Thus the \(m\)-block construction produces a strictly larger
adversarial mean cost than alternation in that range.
\end{proposition}

\begin{proof}
Represent the state as a scalar multiple of \(u\), and let \(z_0\) be its value
immediately before a block of \(m\) copies of \(u\).  After this block and the
following block of \(m\) copies of \(-u\), the state is
\[
 \gamma^{2m}z_0-(1-\gamma^m)^2.
\]
Periodicity gives
\[
 z_0=-\frac{1-\gamma^m}{1+\gamma^m}.
\]
During the first half of the period, the successive insertion lengths are
\[
 \frac{2}{1+\gamma^m},
 \frac{2\gamma}{1+\gamma^m},
 \ldots,
 \frac{2\gamma^{m-1}}{1+\gamma^m}.
\]
The second half has the same list.  Summing the \(\alpha\)-powers and dividing
by \(2m\) proves Eq.\eqref{eq:block-cost}.

Moreover,
\[
 \frac{C_{m,\alpha}(\gamma)}{C_{1,\alpha}(\gamma)}
 =\frac1m\frac{1-\gamma^{\alpha m}}{1-\gamma^\alpha}
  \left(\frac{1+\gamma}{1+\gamma^m}\right)^\alpha.
\]
For \(m\geq2\), the base in the last factor is larger than one.  The fraction
involving the powers of \(\gamma\) tends to one, so the complete expression
tends to infinity as \(\alpha\) tends to infinity.
\end{proof}

\end{document}